\documentclass[a4paper,leqno,12pt,abstracton]{scrartcl}
\pdfoutput=1
\usepackage{ifthen}
\usepackage{varioref}											% vref ist für die Verweise innerhalb des Dokuments zuständig
\usepackage[T1]{fontenc}							
\usepackage[latin1]{inputenc} %Umlaute (ä, ü, ...) werden automatisch erkannt		
\usepackage[british,UKenglish]{babel}

\usepackage{amssymb}
\usepackage{MnSymbol}
\usepackage{amsmath}
\usepackage{amsthm}
\usepackage{yfonts}
\usepackage[geometry]{ifsym}
\DeclareMathAlphabet{\mathpzc}{OT1}{pzc}{m}{it}
\usepackage{hyperref}
\usepackage{eurosym}
\usepackage{pgf}
\usepackage{tikz}
\usetikzlibrary{arrows,backgrounds,positioning,fit}

\numberwithin{equation}{section}
\swapnumbers
\usepackage{caption}
\addto\captionsngerman{

}

\newtheoremstyle{theorstyle}				%Style defnieren
{5mm}
{4mm}
{}
{}
{\bfseries}
{\normalfont :}
{ }
{}

\theoremstyle{theorstyle}			   		%Style auswählen

\newtheorem{theorem}{Theorem}[section]
\newtheorem{concl}[theorem]{Conclusion}

\newtheorem{prop}[theorem]{Proposition}

\newtheoremstyle{defistyle}				%Style defnieren
{5mm}
{8mm}
{}
{}
{\bfseries}
{\normalfont :}
{ }
{}

\theoremstyle{defistyle}			   		%Style auswählen

\newtheorem{example}[theorem]{Example}

\theoremstyle{plain}						% default 

\usepackage{scrpage2}
\newcommand{\IE}{\mathbb{E}}
\newcommand{\IN}{\mathbb{N}}
\newcommand{\IR}{\mathbb{R}}

\newcommand{\IC}{\mathbb{C}}

\renewcommand\Re{\operatorname{\mathfrak{Re}}}
\renewcommand\Im{\operatorname{\mathfrak{Im}}}

\newcommand{\mathb}{\mathpzc{b}}

\newcommand{\mathm}{\mathpzc{m}}

\newcommand{\mathX}{\mathpzc{X}}

\newcommand{\calF}{\mathcal{F}}

\newcommand{\calL}{\mathcal{L}}
\newcommand{\calM}{\mathcal{M}}
\newcommand{\calN}{\mathcal{N}}

\newcommand{\calU}{\mathcal{U}}
\newcommand{\calX}{\mathcal{X}}

\newcommand{\sgn}{\text{sgn}}

\newcommand{\Res}{\text{Res}}

\newcommand{\Dawson}{\text{F}}
\newcommand{\Hilbert}{\text{H}}

\newcommand{\ii}{\mathrm{i}}

\setkomafont{section}{\large\rmfamily}
\setkomafont{subsection}{\large\rmfamily}
\setkomafont{title}{\rmfamily}

\usepackage[backend=bibtex,style=imsart-nameyear,style=alphabetic]{biblatex}
\usepackage[backend=bibtex,style=imsart-nameyear,style=alphabetic]{biblatex}
\usepackage{csquotes}
\bibliography{Bibliography.bib}
\usepackage{authblk}

\title{\Large Network Structure and \\ Counterparty Credit Risk}
\author{\normalsize \textsc{Alexander von Felbert}\thanks{University of Mannheim, Lehrstuhl fuer Wirtschaftsmathematik I, alexander@mathematik-netz.de}}
\date{\normalsize Munich, June 2015}
\providecommand{\keywords}[1]{\small\textbf{\textsc{Keywords --}} #1}

\begin{document}

\maketitle

	\begin{abstract}
     In this paper we offer a novel type of network model which can capture the \textit{precise structure} of a 
     financial market based, for example, on empirical findings. With the attached stochastic framework it is further possible 
     to study how an arbitrary network structure and its expected \textit{counterparty credit risk} are analytically related to each other. 
     This allows us, for the first time, to \textit{model} the precise structure of an arbitrary financial market and to derive the 
     corresponding expected exposure in a closed-form expression. It further enables us to draw implications for the study of \textit{systemic risk}. 
     We apply the powerful theory of characteristic functions and Hilbert transforms. The latter concept is used to express the characteristic
     function (c.f.) of the random variable (r.v.) $\max(Y, 0)$ in terms of the c.f. of the r.v. $Y$. The present paper applies this concept for the
     first time in mathematical finance. 
     We then characterise Eulerian digraphs as distinguished exposure structures and show that considering the precise network
     structures is crucial for the study of systemic risk. The introduced network model is then applied to study the features of an over-the-counter and
     a centrally cleared market. We also give a more general answer to the question of whether it is more advantageous for the overall counterparty 
     credit risk to clear via a central counterparty or classically bilateral between the two involved counterparties. 
     We then show that the exact market structure is a crucial factor in answering the raised question.
	\end{abstract}

     \keywords{Counterparty credit risk, systemic risk, network structure, network model, analytic function, digraph, graph, Eulerian, 
               characteristic function, Hilbert transform, analytic signal, bilateral \& multilateral netting, advantageousness 
               of a central counterparty.}

\newpage

\section{Introduction}
One risk type that has gained particular attention in recent years, largely due to the credit and financial crisis that started in 2007, is \textit{counterparty credit risk}. On the one hand, over-the-counter (OTC)  markets have been seen to respond heavily to financial distress. On the other hand, centrally cleared markets continued to trade without major disruptions even at the height of the financial crisis\footnote{See, for instance, \cite{Rosenthal2001} or \cite{ECB2009}.}.\ Following the impact of this crisis, the G20 countries therefore decided to thoroughly revamp the OTC derivatives market in 2009 in order to reduce the immanent \textit{systemic risk}. In Europe the reforms are implemented through the so-called \textsl{European Market Infrastructure Regulation}\footnote{See \href{http://eur-lex.europa.eu/LexUriServ/LexUriServ.do?uri=OJ:L:2012:201:0001:0059:EN:PDF}{http://eur-lex.europa.eu/LexUriServ/LexUriServ.do?uri=OJ:L:2012:201:0001:0059:EN:PDF}.} (EMIR). The US equivalent is called the \textsl{Dodd-Frank Act}\footnote{See \href{http://www.gpo.gov/fdsys/pkg/PLAW-111publ203/pdf/PLAW-111publ203.pdf}{http://www.gpo.gov/fdsys/pkg/PLAW-111publ203/pdf/PLAW-111publ203.pdf}.}. At the core of both new regulations is the obligation of the market participants to clear their standard OTC derivatives through a central counterparty (CCP). Non-centrally cleared contracts should be subject to higher capital requirements. These measures are designed to comprehensively change the market structure. Today, many classes of derivatives are already being cleared through CCPs, for example, \textsl{LCH.Clearnet}\footnote{\href{http://www.lchclearnet.com/}{http://www.lchclearnet.com/}} clears interest rate swaps, and \textsl{ICE Clear}\footnote{\href{https://www.theice.com/}{https://www.theice.com/}} or \textsl{CME}\footnote{\href{http://www.cmegroup.com/}{http://www.cmegroup.com/}} clear credit default swaps.\\

Several authors such as \textsc{Nier} et al. \cite{Nier2007}, \textsc{Moussa} \cite{Moussa2011}, \textsc{Rosenthal} \cite{Rosenthal2001} or \textsc{Gai} et al. \cite{Gai2010} emphasise the importance of the precise \textit{market structure} in the context of studying counterparty credit risk and therefore systemic risk. Furthermore, empirical studies have shown that network structures in different countries are quite varied.\footnote{A comprehensive overview of these studies and the used network models is provided by \cite{Cont2013} and \cite{Moussa2011}.} Despite these facts, most previous models in the context of counterparty or systemic risk have assumed a simplistic network structure\footnote{See section 1.3.1 in \cite{Moussa2011} for an overview.} such as complete or star graphs. \textsc{Duffie} \& \textsc{Zhu} \cite{Duffie2011} or \textsc{Cont} \& \textsc{Kokholm} \cite{Cont2011}, for instance, assume a \textit{complete graph}. These simplistic network structures, however, are not able to capture empirical findings such as \cite{Upper2004} or as in chapter 4 in \cite{Moussa2011} and tend to over- or underestimate the overall risk.\\

In section \ref{sec:Model} we present a network model, which is capable of capturing the \textit{precise structure} of \textit{any} given financial market based, for example, on empirical findings. We further introduce a stochastic framework to study how different network structures and counterparty credit risk are analytically connected to each other. This allows us, for the first time, to model the precise structure of an arbitrary financial market and to derive the corresponding expected exposure in a closed-form expression. We take the perspective of a regulator and are mainly interested in the overall risk of a market for a typical day in the future. \ In a first step, we incorporate position uncertainty in size and direction in one single distribution. Our model is then capable of dealing with an \textit{arbitrary graph} as well as accounting for a wide range of distributions that represent a position between two counterparts. In a second step, we use conditional probabilities in order to extend this approach to \textit{arbitrary digraphs}, where the size and the direction of all positions can be determined \textit{independently}. That is, the distribution that represents the position value and therefore the exposure can be chosen independently from the exact structure and thus perfectly adapts to the individual circumstances of a given network. \ To this end, the model only assumes that each non-zero position is distributed identically by an arbitrary symmetrical distribution with an existing mean. We have not incorporated correlations, as suggested by \textsc{Cont} and \textsc{Kokholm} \cite{Cont2011}, because we use the powerful theory of \textit{characteristic functions} (see \textsc{Lukacs} \cite{Lukacs1970}) for analysing sums of \textit{independent} random variables. By using this theory, we deduce how we can analytically capture the process of netting in regards to the associated random variables (r.v.s). Afterwards, we show how to determine the expected credit exposure of a netted position by using the so-called \textit{Hilbert transform} (see \textsc{King} \cite{King2009}). For that purpose, we use expressions for the c.f. of the random variable $\max(Y, 0)$ in terms of the c.f. of the r.v. $Y$ based on results of \textsc{Pinelis} \cite{Pinelis2013}. The present paper applies this concept for the first time in mathematical finance.
Nevertheless, Hilbert transform methods have been used before in mathematical finance, for instance, by \textsc{Feng} and \textsc{Linetsky} \cite{Feng2008} to price discretly-monitored single- and double-barrier options in Levy process-based models. An overview of Fourier transform methods in finance can be found in \cite{Cherubini2010}.\\

In section \ref{sec:AuxiliaryResults} we provide auxiliary results which can be used for the application of the network model. In the first subsection we prove Proposition \ref{prop:Hilbert}, which contains two very useful formulas about Hilbert transforms, by using the residue theorem. These formulas are particularly useful for calculating the Hilbert transform of intricate functions. We also study the so-called positive and negative absolute values of a distribution in section \ref{subsec:PositiveNegativePart}. Both types are used to represent the direction of a position. The term 'analytic signal', known from the field of \textit{signal processing}, is then introduced, and we show in Proposition \ref{prop:AnalyticSignal} that the positive absolute value of a distribution is an analytic signal. \ Some of these insights are used in section \ref{sec:StructureTheorem} to prove both structure theorems \ref{theorem:digraphs} and \ref{theorem:graphs}. The theorems basically state that Eulerian digraphs are distinguished exposure structures and that digraphs possess different characteristics compared to graphs in the context of counterparty credit risk. We further reveal that different structures within graphs or digraphs can have a significantly different impact on the overall counterparty risk.\\ 

We then apply our network model and its stochastic framework in section \ref{sec:Application} to study the features of bilateral and multilateral clearing and to give a more general answer to the question raised by \textsc{Duffie} \& \textsc{Zhu} \cite{Duffie2011}, of whether it is more advantageous for the overall counterparty risk to clear via a CCP or classically bilateral between the two involved counterparties. \ The two authors model the counterparty credit risk of \textit{each} market participant and for both netting types as an independent and standard normal distributed random variable in order to answer this question. This web of obligations and claims, described in \cite{Duffie2011}, can be illustrated as a  complete graph. With the introduced network model we can answer this question not only for complete graphs, but for arbitrary graphs and digraphs as well. Moreover, the network model introduced in section \ref{sec:Model} is not constrained to the normal distribution. It can also employ any (symmetric) distribution with a defined expected value. We finally show in section \ref{subsec:Advantageousness} that the question of the advantageousness of one netting type also depends heavily on the precise structure of the market by comparing the implications of our model and the model used by \cite{Duffie2011}.

\section{Network Model for a Financial Market}
\label{sec:Model}
\textit{Counterparty credit risk}, often known just as \textit{counterparty risk} or \textit{default risk}, is usually defined as the risk that the entity with whom one has entered into a financial contract will fail to fulfill his side of the contractual obligations.\footnote{This kind of risk arises in almost every financial market such as the derivatives market, the interbank market, the money market, or the repo market.} \textit{Credit exposure} or simply \textit{exposure} defines the actual loss in the event of a counterparty defaulting. \ In the next two subsections we explain the basic settings for a general financial market which is subject to counterparty risk. We start with a market modelled as a graph where size and the direction of a position is determined by a random variable. We extend this model by using conditional probabilities in order to determine the direction well before any observation is drawn. This allows us to model the exact directed network structure of any financial market, where only the position size is a matter of coincidence. Furthermore, we introduce the stochastic framework and a set of formulas to calculate the expected credit exposure.
\subsection{Market Settings}
\label{subsec:MarketSettings}
% very basics about the financial market: conventions, market participants, derivatives classes
% how to interpret a graph and how to model the different classes
We consider a financial market $\calM$ with $N\in \IN$ \textit{participants} and $K\in \IN$ different \textit{classes of derivatives} $C:=\{1, \ldots, K\}$. Derivatives classes could be defined by underlying asset classes, but we could also aggregate different underlying asset classes to one derivatives class. Let $k\in C$ and $m_k\in \IN$. We model each financial sub-market of derivatives class $k$ as a single \textit{graph} $G_k=(V=\{v_1, \ldots, v_N\}, \ E^k=\{e^k_1, \ldots, e^k_{m_k}\})$. It consists of a non-empty finite set $V = V (G_k)$ of elements called \textit{vertices} and a finite set $E^k = E(G_k)$ of unordered pairs of distinct vertices called \textit{edges}. The vertices of a given graph $G_k$ represent the $N$ \textit{market participants}, and the edges of $E^k$ stand for the \textit{trade positions} or simply \textit{positions} between two different counterparts within derivatives class $k\in C$. A \textit{trade position} is the net value of a bilateral portfolio within derivatives class $k\in C$. \ Furthermore, a financial market $\calM$ is usually endowed with a set of \textit{market conventions}, that apply to each of the $N$ participants within a class of derivatives. For instance, the type of netting or the day-count conventions are typical market conventions. %
%how do we model the entire financial market
We write $\widehat{E}:= \bigcupdot_{k \in C}{E^k}$ for the \textit{compounded set of edges}\footnote{This is, for the sake of simplicity, an abuse of notation.} and $G:=(V, \widehat{E})$ shall represent all graphs $G_k$ on the common set of vertices $V$. Here, $\bigcupdot$ stands for the disjoint union of sets.\\

%  assumptions: exposure distribution (symmetric, zero mean), deterministic part, current position values are available
%  modelling the uncertainty of future trade position by random variables
We assume that the uncertainty of the value of a future bilateral trade position can be represented by a real-valued \textit{probability distribution} $P$, which is symmetric around the origin and with zero mean. That is, we model the uncertainty of size and direction of a position 
of a counterpart $v\in V$ relative to counterpart $w\in V\setminus \{v\}$ in derivatives class $k\in C$ by a r.v. $X^{(k)}_{v, w} \sim P$. A \textit{realisation} of the r.v. $X^{(k)}_{v, w}$ is denoted by $x^{(k)}_{v, w}\in \IR$. If $x^{(k)}_{v,w}$ is positive then $v$ will claim this amount from $w$, but if $x^{(k)}_{v,w}$ is negative then $v$ will owe the amount of $x^{(k)}_{v,w}$ to $w$. The \textit{direction} as well as the associated \textit{size} or \textit{weight} of an edge $\{v, w\}\in E^k$ is then defined by the observation $x^{(k)}_{v,w}$ of the random experiment.\\

% how to interpret a digraph and how to model the different classes
The given graph $G_k$ supplemented by the directions of each of its corresponding realisation $x^{(k)}_{v,w}$ represents a so-called \textit{directed graph} or just \textit{digraph} $D_k=(V, A^k=\{a^k_1, \ldots, a^k_{m_k}\})$. $D_k$ also comprises the vertex set $V$ and a set $A^k \subseteq (V \times V)$ of ordered pairs of different vertices called \textit{arrows}, as well as two maps $h:A \rightarrow V$ and $t:A \rightarrow V$ assigning to every arrow $a\in A$ a \textit{head vertex} $h(a)$ and a \textit{tail vertex} $t(a)$. Within a digraph we know that the creditor $h(a)\in V$ claims the position value from the debtor $t(a)\in V$ for all $a\in A^k$. A digraph $D_k=(V, A^k)$ is called an \textit{orientation} of a graph $G_k=(V, E^k)$, if each edge $\{v, w\}\in E^k$ is replaced by one of the ordered pairs $(v, w)$ or $(w, v)$, i.e., the digraph $D_k$ along with its realisations is an orientation of $G_k$. \\

We call a graph or digraph \textit{weighted}, or \textit{network}, if each link of each derivatives class $k$ is assigned with a real number. A realisation $x^{(k)}_{v,w}\in \IR$ of the r.v. $X^{(k)}_{v,w}$ also provides a weight\footnote{Also called size within this document.} for each edge of $G_k$. That is, the outcomes of the r.v.s $X^{(k)}_{v,w}$ define a weighted orientation of $G_k$. Analogously, we write $\widehat{A}:= \bigcupdot_{k \in C}{A^k}$ for the \textit{compounded set of arrows} and $D:=(V, \widehat{A})$ shall represent all digraphs $D_k$ with $k\in C$ on the common set of vertices $V$.\\

%  definition of the conditional r.v. distributed by P and  +/-|P|
If we only want to model the weights of a digraph with pre-defined directions by r.v.s, we can use the following technique. We start over with an undirected graph $G_k=(V, E^k)$. Each r.v. $X_{v, w}^{(k)} \sim P$ shall be supplemented with a direction in form of a condition, denoted by $\pm$,  defining whether the direction shall be positive or negative. We set
\begin{equation*}
	{}^\pm|X^{(k)}_{v,w}| \ := \  
     \left\{
	\begin{aligned}
		+|X^{(k)}_{v,w}| \ , \quad  \text{if } (w, v)\in A^k \\
		-|X^{(k)}_{v,w}| \ , \quad  \text{if } (v, w)\in A^k 
	\end{aligned} \ 
	\right. 
\end{equation*}
and say that $+|X^{(k)}_{v,w}|=|X^{(k)}_{v,w}|$ is the \textit{positive absolute value} and $-|X^{(k)}_{v,w}|$ the \textit{negative absolute value} of the distribution $P$ or of the r.v. $X^{(k)}_{v,w}$. That is, we take the absolute value of the symmetric r.v. $X^{(k)}_{v,w}$ and multiply the outcome by $\pm 1$. Here, the factor $+1$ stands for an incoming and $-1$ for an outgoing arrow of the vertex $v\in V$. The edge becomes an arrow, and the associated distribution, denoted by ${}^{\pm}|X^{(k)}_{v,w}| \sim {}^{\pm}|P|$, is not symmetric anymore. It contains either only positive or only negative outcomes. For instance, the probability density function $f$ of the r.v. $X^{(k)}_{v,w} \sim \calN(0,1)$ and the density function $f^+$ of its  positive absolute value $|X^{(k)}_{v,w}|$ is given by $\sqrt{\frac{2}{\pi}}e^{-\frac{x^2}{2}}$ for $x>0$ and $0$ otherwise. We denote the \textit{normal distribution} by $\calN(\mu, \sigma)$, where $\mu$ is the mean and $\sigma$ the standard deviation. Because of the definition of $|X^{(k)}_{v,w}|$, we already know that $v$ claims the trade position from $w$ before any observation is drawn from $|X^{(k)}_{v,w}|$. If we apply this technique to all edges it allows us to study any orientation $D_k$ of the given graph $G_k$. \ In the following we will, for the sake of simplicity, first introduce the market settings for directed graphs. Afterwards we adapt the notation to undirected graphs simply by forgetting the direction of the arrows or positions.

\subsection{Closeout Netting and Credit Exposure}
\label{subsec:NettingExposure}
%
%  definition closeout netting agreement
So-called \textit{closeout netting agreements} have been a common tool for reducing the credit exposure of an entire market $\calM$. \  A closeout netting agreement is a legally binding contract between two parties. It stipulates that if one counterparty defaults, legal obligations arising from derivative transactions covered by the netting agreement must be based on the net value of such transactions. We do not consider benefits of collateral and default recovery.\\

%  definition netting set
The applied closeout netting convention of the market $\calM$ defines exactly what trades of an arbitrary counterparty $v\in V$ can be aggregated into one net position in the case of a default. Expressed in terms of set theory, this means that the applied netting convention defines a partition\footnote{See \cite{Diestel2005}, Chapter 1 for the general definition of a partition of a given set.}% 
$\calL_v$ of all links that are incident\footnote{The vertex $v$ of a graph $G_k$ is \textit{incident} with an edge $e\in E^k$ if $v\in e$. We further call an arrow $a\in A^k$ of the digraph $D_k$ \textit{incident} with the vertex $v$ if $h(a)=v$ or $t(a)=v$.} %
to $v$. That is, $\widehat{A}(v):=\{a\in A^k \ | \ v \text{ incident with } a; \ k\in C \}$ can be decomposed into $\widehat{A}(v) \ = \ \bigcupdot_{\Lambda \in \calL_v}{\Lambda}$ with $\Lambda\neq \emptyset$, where all trades of each set $\Lambda\in \calL_v$ have $v$ as one of the two counterparts in common. We call $\Lambda\in \calL_v$ a \textit{netting set} of the market participant $v$.\\

%  connection of the netting set to market conventions and example
Obviously, a netting set is strongly dependent on the used netting opportunity as part of the market conventions. Within an OTC market, for instance, the \textsl{ISDA Master Netting Agreement}\footnote{See \href{http://www2.isda.org/}{http://www2.isda.org/}.} is a standard closeout netting agreement which allows two bilateral counterparts to net across different kinds of derivatives. Although netting across all product categories is often not allowed\footnote{See section 3.4.7 in \cite{Gregory2010}.}, we will in the following net across all classes of derivatives of $C$. In contrast, a CCP offers the possibility of netting across all its clearing members\footnote{See section 3.4.10 and 14.1 in \cite{Gregory2010}.}. Bilateral as well as multilateral netting will be treated in detail in sections \ref{sec:BilateralNetting} and \ref{sec:MultilateralNetting}.\\ 

%  definition of the set of r.v. for a netting set
Each netting set $\Lambda\in \calL_v$ corresponds with a set of r.v.s $\mathX_{\Lambda}:= \ \{{}^\pm|X_\lambda|\}_{\lambda \in \Lambda}$, where each element of a netting set represents the future value of a bilateral trade position with counterpart $v$. If we want to calculate the expected counterparty risk of a set of r.v. $\{{}^\pm|X_\lambda|\}_{\lambda \in \Lambda}$ of a market participant $v$, we need to specify whether an arrow $\lambda$ is a claim or a debt of $v$ relative to its counterpart. For this purpose we use the already introduced notation $X_{v, w}^{(k)}$. Further we write $\Lambda_v$ to stress that each trade position is meant relative to $v$, i.e., if the position is positive then $v$ claims the amount from $w$, and if it is negative then $v$ owes the amount to $w$. We designate $\mathX_{\calL_v} := \ \bigcup_{\Lambda \in \calL_v}{ \mathX_\Lambda}$ as the \textit{family of sets of r.v.s of} $v$ that implies the counterparty risk of the market participant $v$ for the entire market.

\begin{figure}[htb]
\centering
\tikzstyle{blueStyle}=[->, circle,draw=black!80,fill=gray!30,thick,inner sep=0pt,minimum size=4mm]
\tikzstyle{transparentStyle}=[->, circle,draw=white!100,fill=white!100]
\tikzstyle{pre}=[<-,shorten <=1pt,>=stealth',semithick]
\tikzstyle{post}=[->,shorten >=1pt,>=stealth',semithick]
\tikzstyle{every loop}=[min distance=15mm, looseness=1.5,out=30, in=-30]

\begin{tikzpicture}[scale=0.75]
\node[blueStyle] (v1) at (0,0) {$v_1$};
\node[blueStyle] (v2) at (3,0) {$v_2$};
\node[blueStyle] (v3) at (3,3) {$v_3$};
\node[blueStyle] (v4) at (0,3) {$v_4$};

%\node[transparentStyle] (d1) at (2,0.2) {\Large$\vdots$};
\draw[->, black!100, very thick, dashed] (v1) to[] node[below, font=\footnotesize] {$a^{2}_2$} (v2);
\draw[<-, black!100, thick] (v1) to node[right, font=\footnotesize] {$a^{1}_1$} (v3);
\draw[<-, black!100, very thick, dashed] (v2) to[bend right=45] node[right, font=\footnotesize] {$a^{2}_1$} (v3);
\draw[->, black!100, thick] (v2) to[] node[right, font=\footnotesize] {$a^{1}_4$} (v3);
\draw[->, black!100, thick] (v3) to node[above, font=\footnotesize] {$a^{1}_2$} (v4);
\draw[->, black!100, very thick, dashed] (v4) to[bend right=45] node[left, font=\footnotesize] {$a^{2}_3$} (v1);
\draw[->, black!100, thick] (v4) to[] node[left, font=\footnotesize] {$a^{1}_3$} (v1);
\end{tikzpicture}
\caption{Digraph $D=(V:=\{v_1, v_2, v_3, v_4\}, \widehat{A})$ with $\widehat{A}=\{a_1^{1}, a_2^{1}, a_3^{1}, a_4^{1}, a_1^{2}, a_2^{2}, a_3^{2}\}$.}
\label{abb:GraphAndPartition}
\end{figure}
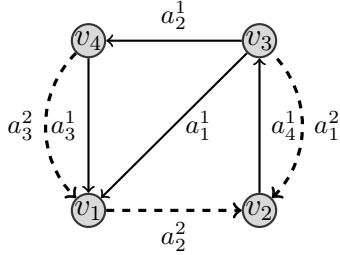

For instance, an evident partition of the set $\widehat{A}(v_1)$, as graphically suggested in Fig. \ref{abb:GraphAndPartition} by the different types of arrows, is $\widehat{A}(v_1)= \{a^{1}_1, a^{1}_3\} \cupdot \{a^{2}_2, a^2_3\}$. The netting set $\Lambda_{v_1} = \{a^{2}_2, a^{2}_3\}$ of Fig. \ref{abb:GraphAndPartition}, for example, corresponds with the set $\mathX_{\Lambda_{v_1}}=\{|X^{(2)}_{v_4, v_1}|, -|X^{(2)}_{v_1,v_2}|\}$ of random variables. The partition $\calL_{v_1} = \{ \ \{a^{1}_1, a^{1}_3\}, \{a^{2}_2, a^2_3\} \ \}$ contains all netting sets of the counterpart $v_1$ and the family of sets of r.v.s $\mathX_{\calL_{v_1}}$ of $v_1$ reflects the uncertainty of the future value of the positions between $v_1$ and its counterparts.\\ 

The objective of this very general notation of netting is to be as flexible as possible so that the model can cope with arbitrary netting types. In addition, we do not need to handle huge adjacency matrices, and we can apply the introduced notation for the next definition. We say that a digraph $D$ \textit{is distributed by} $P$, denoted by $\mathX \sim {}^\pm|P|$, if and only if all trade positions of $\mathX_{\calL_v}$ for all $v\in V$ are distributed independently and identically (i.i.d.) by ${}^\pm|P|$. We adapt a similar notation $\mathX \sim P$ for undirected graphs $G$.\\

After knowing the relevant netting sets and defining how to interpret the r.v.s, netting is then simply performed by adding the estimated future position values in form of the random variables. By taking the maximum between the netted sum and zero we determine the credit exposure for the counterparty $v$ and netting set $\Lambda_v$. Hence, the counterparty risk of $v$ considering the netting set $\Lambda_v$ is determined by
\begin{align}
     \label{form:MaxOfSum}
   \max\left[\sum_{\lambda \in \Lambda_v}{{}^\pm|X_\lambda|} ; \ 0 \right],
\end{align}
where ${}^\pm|X_\lambda|$ is the positive or negative absolute value of a real-valued symmetric r.v. $X_\lambda \sim P$ with zero mean.\ The theory of characteristic functions is a powerful tool for analysing sums of independent random variables. If the r.v.s of the finite sequence $(X_\lambda)_{\lambda\in \Lambda}$ are mutually independent then the c.f. of the sum $Y:=\sum_{\lambda \in \Lambda}{X_\lambda}$ is simply the product 
\begin{align}
	\label{form:SumRV}
	\phi_Y = \prod_{\lambda \in \Lambda}{ {}^\pm|\phi_{X_{\lambda}}| } 
\end{align}
of the corresponding characteristic functions. In general the function $\phi_Y$ is complex-valued, that is, $\phi_Y(t) = \eta(t)+ \ii \nu(t)$ with real part $\Re(\phi_Y)=\eta$ and imaginary part $\Im(\phi_Y)=\nu$. \ Apparently if we want to apply this concept to (\ref{form:MaxOfSum}), we have to find a way to figure out the c.f. of the positive as well as the negative absolute value of the given distribution $P$. We derive formulas for this purpose in section \ref{subsec:PositiveNegativePart}.\\

%%% Only the positive absolute value of the netted exposure matters

%% Maximum(sum, 0) = Exposure
If a market participant $w$ claims money from the defaulted counterpart $v$, then $w$ will probably incur a loss. Whilst if $w$ owes money to the defaulted counterpart $v$, then $w$ will still have to honour the contractual payments. That is, in the latter situation $w$ cannot gain from the default by being somehow released from their liability. Thus, only a positive trade position implies an exposure greater than zero. The exposure can be figured out by using formula (4) in \cite{Pinelis2013} in order to obtain the c.f. 
\begin{align}
	\label{form:MaxOfRV}
	\phi_{\max[Y; 0]}(t) \ = \ \IE(e^{\ii t \max[Y; 0]}) \ = \ \frac{1}{2}[1+\phi_{Y}(t)] +  \frac{\ii}{2}[\Hilbert\{\phi_{Y}\}(t) - \Hilbert\{\phi_{Y}\}(0)]
\end{align}
of the r.v. $\max[Y; 0]$. Here, $\ii$ is the imaginary unit and $\Hilbert\{\phi_Y\}$ is the Hilbert transform\footnote{See section 3.1 in \cite{King2009}.} of the (characteristic) function $\phi_Y$ given by 
\begin{align}
	\label{form:defHT}
	\Hilbert \{\phi_Y(t)\}(\omega) 
     := \ \frac{1}{\pi} PV \int_{-\infty}^{\infty}{ \frac{\phi_Y(t) \,dt}{\omega-t}} 
     := \ \lim_{\epsilon \rightarrow 0} \frac{1}{\pi} \left( \int_{-\infty}^{\omega-\epsilon}{ \frac{\phi_Y(t) \,dt}{\omega-t}} + \int_{\omega + \epsilon}^{\infty}{ \frac{\phi_Y(t) \,dt}{\omega-t}} \right)  
\end{align}
with $t, \omega \in \IR$ and provided this integral exists. The $PV$ in front of the integral denotes the \textit{Cauchy principal value}\footnote{See section 2.4 in \cite{King2009}.} that expands the class of functions for which the ordinary improper integral exists. When it is clear from the context what is meant we will use the variable $t$ for the argument of the input function as well as for the argument of its Hilbert transform. According to Theorem 2.3.1 and its Corollary 2 in \cite{Lukacs1970} we can derive the expected value of any r.v. $Z$ by 
\begin{align}
	\label{form:CharFuncMoments}
	\IE(Z) \ = \ \ii^{-1} \partial_t[\phi_{Z}(t)](0),
\end{align}
on the condition that the first moment exists. In our case, we set $Z:=\max[Y; 0]$ in order to compute the desired expectation.

\subsection{An Illustrative Example}
\label{subsec:Example}

Several authors use simplified network structures such as complete, star, or random graphs. For instance, \textsc{Duffie} et al. \cite{Duffie2011} or \textsc{Cont} et al. \cite{Cont2011} assume complete graphs. Other authors such as \textsc{Nier} et al. \cite{Nier2007} assume that the edge set follows the Erdoes-Renyi model. The present model, however,  can deal with arbitrary graphs.
\begin{example}
Consider the financial market $\calM$ with four market participants in which interest rate and FX derivatives are traded. The market is depicted in Fig. \ref{abb:IllustrativeGraph} with  $G_k=(V, \ E^k)$ with $k\in C:=\{1, 2\}$ and $V:=\{v_1, v_2, v_3, v_4\}$. Here, the edges of $E^1=\{e^1_1, e^1_2, e^1_3, e^1_4\}$ and $E^2=\{e^2_1, e^2_2, e^2_3\}$ are represented by solid and dashed lines, respectively. Apparently neither of the two graphs $G_1$ and $G_2$ is complete or of star form, and both are different.
\begin{figure}[htb]
\centering
\tikzstyle{blueStyle}=[-, circle,draw=black!80,fill=gray!30,thick,inner sep=0pt,minimum size=4mm]
\tikzstyle{transparentStyle}=[-, circle,draw=white!100,fill=white!100]
\tikzstyle{pre}=[-,shorten <=1pt,>=stealth',semithick]
\tikzstyle{post}=[-,shorten >=1pt,>=stealth',semithick]
\tikzstyle{every loop}=[min distance=15mm, looseness=1.5,out=30, in=-30]
\begin{tikzpicture}[scale=0.9]
\node[blueStyle] (v1) at (0,0) {$v_1$};
\node[blueStyle] (v2) at (3,0) {$v_2$};
\node[blueStyle] (v3) at (3,3) {$v_3$};
\node[blueStyle] (v4) at (0,3) {$v_4$};

%\node[transparentStyle] (d1) at (2,0.2) {\Large$\vdots$};
\draw[-, black!100, very thick, dashed] (v1) to[] node[below, font=\small] {$e^{2}_2$} (v2);
\draw[-, black!100, thick] (v1) to node[right, font=\small] {$e^{1}_1$} (v3);
\draw[-, black!100, very thick, dashed] (v2) to[bend right=45] node[right, font=\small] {$e^{2}_1$} (v3);
\draw[-, black!100, thick] (v2) to[] node[right, font=\small] {$e^{1}_4$} (v3);
\draw[-, black!100, thick] (v3) to node[above, font=\small] {$e^{1}_2$} (v4);
\draw[-, black!100, very thick, dashed] (v4) to[bend right=45] node[left, font=\small] {$e^{2}_3$} (v1);
\draw[-, black!100, thick] (v4) to[] node[left, font=\small] {$e^{1}_3$} (v1);
\end{tikzpicture}
\caption{Digraph $D=(V, \widehat{A})$ and underlying graph $G=(V, \widehat{E})$}
\label{abb:IllustrativeGraph}
\end{figure}
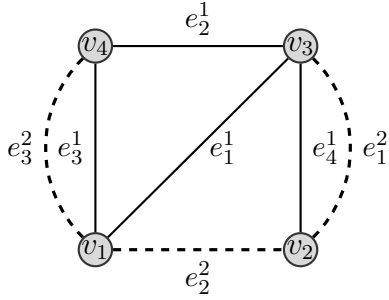
We further assume that $\calX \sim \calL(0, 1)$ where $\calL(\mu, b)$ denotes the \textit{Laplace distribution} with mean $\mu$ and scaling parameter $b$. The  partitions 
\begin{align*}
   \calL_{v_1} = \{ \ \{e_1^1, e_3^1\}, \{e_2^2\}, \{e_3^2\} \ \}     \qquad \calL_{v_2} &= \{ \ \{e_4^1\}, \{e_1^2\}, \{e_2^2\} \ \} \\
   \calL_{v_3} = \{ \ \{e_1^1, e_2^1, e_4^1\}, \{e_1^2\} \ \}         \qquad \calL_{v_4} &= \{ \ \{e_2^1, e_3^1\}, \{e_3^2\} \ \}
\end{align*}
and their elements, the netting sets, are determined by the netting type\footnote{In anticipation of section \ref{sec:Application}, multilateral and bilateral netting is applied to the edge sets $E^1$ and $E^2$, respectively.}. %
According to formulas (\ref{form:SumRV}) - (\ref{form:CharFuncMoments}) the expected counterparty risk of a netting set $\Lambda \in \calL_{v_i}$ with $i\in \{1,2,3,4\}$ can be calculated in four steps:
\begin{enumerate}
 \item[(a)] Determine the c.f. $\phi_Y$ of the netted position by using formula (\ref{form:SumRV});
 \item[(b)] Calculate the Hilbert transform of $\phi_Y$;
 \item[(c)] Determine the c.f. $\phi_{\max(Y;0)}$ of the credit exposure with formula (\ref{form:MaxOfRV});
 \item[(d)] Apply formula (\ref{form:CharFuncMoments}) to get the expected credit exposure.
\end{enumerate}
Afterwards the expected credit exposures of netting sets can be summed up because of their additivity.\\ 

We start with the c.f. $\phi_X(t)= \frac{1}{1+t^2}$ of a single r.v. $X \sim \calL(0,1)$ and we chose $\Lambda \in \calL_{v_i}$ with $|\Lambda|=1$. By applying formula (\ref{form:SumRV}) and considering that $X=Y$ we obtain $\phi_Y=\phi_X$ for all seven single-element netting sets. The calculation of the Hilbert transform $\Hilbert\{\phi_{Y}\}(t)= \frac{t}{1+t^2}$ is straightforward\footnote{Please refer to Example \ref{example:LaplaceMultiNetting}.} and the c.f. 
\begin{align*}
   \phi_{\max[Y; 0]}(t) \ = \ \frac{1}{2}\left[1+\phi_{Y}(t)\right] + \frac{\ii}{2}\left[\Hilbert\{\phi_{Y}\}(t) - \Hilbert\{\phi_{Y}\}(0)\right] = \frac{1}{2}\left[1+\frac{1}{1+t^2}\right] + \frac{\ii}{2}\left[\frac{t}{1+t^2}\right]
\end{align*}
of the credit exposure can be determined by formula (\ref{form:MaxOfRV}). Finally, we get $\IE(\max[Y; 0]) \ = \ \ii^{-1} \partial_t[\phi_{\max[Y; 0]}(t)](0) = \frac{1}{2}$ for the expected credit exposure of $\Lambda$ by applying formula (\ref{form:CharFuncMoments}).\\

For the remaining three netting sets the same steps (a)-(d) need to be performed, but the r.v. $Y$ and the c.f. $\phi_Y$ will be different due to the bigger netting set. The netting set $\Lambda'=\{e_1^1, e_2^1, e_4^1\}$, for instance, contains three edges which implies that $\phi_Y(t) = \phi_X^3(t) = \frac{1}{(1+t^2)^3}$. Thus, $\Hilbert\{\phi_Y\}(t)=\frac{15t}{8(1+t^2)^3}+\frac{5t^3}{4(1+t^2)^3}+\frac{3t^5}{8(1+t^2)^3}$ which leads to $\IE(\max[Y; 0]) \ = \ \frac{15}{16}$. Each of the two-element netting sets entails an additional $\frac{3}{4}$ of expected counterparty risk. Adding up, we receive $\frac{95}{16}=7 \times \frac{1}{2} + 2\times \frac{3}{4} + \frac{15}{16}$ as the expected counterparty risk of the entire market. 
\end{example}

The models of \cite{Duffie2011} or \cite{Cont2011} would have assumed that both graphs are complete, which means that each participant is connected via a trade position to all others. Another huge advantage of the present network model is that it can cope with a wide range of distributions. In the example above we have used the Laplace distribution but we could also have applied any other symmetric distribution with existing mean.\\

In situations where we have additional information about the possible direction of an exposure, it can be reasonable to model directed exposure. This means that only the position size is a matter of coincidence. Our network model can deal with arbitrary digraphs as well by taking the additional information in the c.f. into account. That is, for digraphs we need to calculate the c.f. of a single random variable, for instance, according to Proposition \ref{prop:AnalyticSignal}. Afterwards the remaining steps (b)-(d) are identical. In Example \ref{exa:BilTwoTierStructure} the expected counterparty risk of a directed simple two-tier market structure is calculated.\\

As long as we are capable of determining the related c.f.s and the corresponding Hilbert transforms we can use the outlined network model and the attached stochastic framework to determine the expected counterparty risk of an arbitrary graph or digraph. In section \ref{sec:AuxiliaryResults} we present auxiliary results, how some of the hurdles can be overcome in performing steps (a)-(d).

\section{Specific Network Structures and Counterparty Risk}
\label{sec:StructureTheorem}
To deduce a formula for the expected counterparty credit risk of an arbitrary network structure and for such a wide range of possible distributions, even without considering dependencies between the different positions, is a demanding task. Let us assume that the digraph $D=(V, \widehat{A})$ represents a network structure and that future positions are distributed by $\mathX \sim {}^\pm|P|$. The challenge is then to compute $\IE\left(\max\left[\sum_{\lambda \in \Lambda_v}{{}^\pm|X_\lambda|}, 0\right]\right)$ for an arbitrary netting set $\Lambda_v$ of a counterpart $v$ and to deal with various of problems attached to it: The negative as well as the positive absolute value is not distributed by $P$ anymore, because their sample space is restricted either to $]-\infty, 0[$ or to $]0, \infty[$. The probability distribution of the sum $\sum_{\lambda \in \Lambda}{{}^\pm|X_\lambda|}$ is actually the convolution of their distributions, and in general, little can be said about it. Finally, taking the maximum causes some sort of asymmetry of the problem, and this implies the non-additivity as the following example illustrates.

\begin{example}
Suppose $D=(V, A=\{a_1, a_2\})$ is the digraph as depicted in Fig. \ref{abb:P3} and represents a market with $K=1$ class of derivatives. 

\begin{figure}[htb]
\centering
\tikzstyle{blueStyle}=[->, circle,draw=black!80,fill=gray!30,thick,inner sep=0pt,minimum size=4mm]
\tikzstyle{pre}=[<-,shorten <=1pt,>=stealth',semithick]
\tikzstyle{post}=[->,shorten >=1pt,>=stealth',semithick]
\tikzstyle{every loop}=[min distance=15mm, looseness=1.5,out=30, in=-30]

\begin{tikzpicture}
\node[blueStyle] (u) at (0,0) {$u$};
\node[blueStyle] (v) at (1.5,0) {$v$};
\node[blueStyle] (w) at (3,0) {$w$};
\draw[->, black!100, thick] (u) to node[below, font=\small] {$a_1$} (v);
\draw[->, black!100, thick] (v) to node[below, font=\small] {$a_2$} (w);
\end{tikzpicture}
\caption{Path $D$}
\label{abb:P3}
\end{figure}
The associated r.v.s $X_{a_1}$ and $X_{a_2}$ are distributed i.i.d. by the \textit{continuous uniform distribution} $\calU(\{-1, 1\})$. The vertex $u$ has no expected counterparty risk at all, as $\IE(\max[-|X_{a_1}|; 0])=\IE(0)= 0$. The end-vertex $w$ of $D$ obviously entails $\IE(\max[|X_{a_2}|; 0])=\IE(|X_{a_2}|)= \frac{1}{2}$ expected counterparty credit risk. \ The c.f. of the sum $Y:=|X_{a_1}|-|X_{a_2}|$ equals $\phi_{\calU(0,1)}(t)\phi_{\calU(-1,0)}(t) = \frac{(1-e^{-\ii t})(-1+e^{\ii t})}{t^2}$ and its Hilbert transform $\Hilbert\{\phi_{Y}\}$ is $\frac{2(t-\sin(t))}{t^2}$. Taking the limit $\lim_{t \rightarrow 0} \Hilbert\{\phi_Y\}(t)$ and applying formulas (\ref{form:MaxOfRV}) to (\ref{form:CharFuncMoments}) we obtain $\IE\left(\max[Y; 0 ]\right) \ =  \frac{\partial_t\phi_{\max[Y; 0 ]}(0)}{\ii}=\frac{1}{6}$. Please note that $\IE\left(\max[|X_{a_1}|-|X_{a_2}|; 0 ]\right) \neq \IE\left(\max[|X_{a_1}|; 0 ]\right)+\IE\left(\max[-|X_{a_2}|; 0 ]\right)$. Economically, this inequality means that a breakdown of the credit exposure of one participant into smaller pieces is only possible by respecting the netting rules. Here, the multilateral netting rules are not respected by splitting up the sum $|X_{a_1}|-|X_{a_2}|$. 
\end{example}
In addition, the example demonstates that the more liabilities a counterpart $v$ has relative to its claims, the lower the counterparty credit risk will be for $v$. Considering the netting efficiency of a single counterparty the situation changes: the better the balance between claims and liabilities, the greater the offsetting effect of the netting opportunity.
\begin{figure}[htb]
\centering
\tikzstyle{blueStyle}=[->, circle,draw=black!80,fill=gray!30,thick,inner sep=0pt,minimum size=4mm]
\begin{tikzpicture}[scale=0.7]
\node[blueStyle] (u) at (0,0) {$u$};
\node[blueStyle] (v) at (2,0) {$v$};
\node[blueStyle] (w) at (0,2) {$w$};
\draw[->, black!100, thick] (u) to node[below, font=\small] {} (v);
\draw[->, black!100, thick] (v) to node[above, font=\small] {} (w);
\draw[->, black!100, thick] (w) to node[left, font=\small] {} (u);
\end{tikzpicture}
\caption{Exposure circle}
\label{abb:ExposureCircle}
\end{figure}
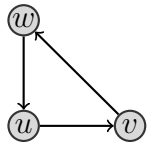
In anticipation of section \ref{sec:MultilateralNetting}, a popular example for the netting efficiency of a centrally cleared market is shown in Fig. \ref{abb:ExposureCircle}. Let us assume that each arrow of the exposure circle represents an exposure of \EUR{100} million. Then, the circle of exposure implies a perfect balance between the claims and liabilities of each participant, because the in- and outgoing arrows offset each other completely. If we generalise this obvious concept of exposure circles and use the language of graph theory we come across Eulerian digraphs. If we further replace the deterministic values by r.v.s that represent the future counterparty credit risk between two participants, then we come to Theorem \ref{theorem:CounterpartyRiskDigraph}. \ To be able to do so, however, we need to introduce the \textit{degree} $\gamma(v)$ of a vertex $v\in V$ within a graph $G_k=(V, E^k)$, which is the number $|E(v)|$ of different edges at $v$. Let us now consider a digraph $D_k=(V, A^k)$. The \textit{in-degree of} a single vertex $v$, denoted by $\gamma_+(v)$, is then the number of arrows $a\in A^k$ with $h(a)=v$. Similarly, we call the number of arrows $a\in A^k$ with $t(a)=v$ the \textit{out-degree of} $v$ and denote it by $\gamma_-(v)$. We shall call $\gamma_+: v \mapsto \gamma_+(v)$ the \textit{in-degree function}  and $\gamma_-: v \mapsto \gamma_-(v)$ the \textit{out-degree function}. Moreover, we define $\gamma(v):= \gamma_+(v) - \gamma_-(v)$ for a digraph $D_k$ and call $\gamma$ the \textit{Eulerian degree function} and $\gamma(v)$ the \textit{Eulerian degree of} $v$.
\begin{theorem}
\label{theorem:digraphs}
Let $D_k=(V, A^k)$ be a connected digraph with $\mathX \sim {}^\pm|P|$ and a netting set $\Lambda_v$ with $v\in V$. Then the following holds: 
\begin{enumerate}
	\item[(i)] $\IE(\sum_{\lambda \in \Lambda_v}{{}^\pm|X_\lambda|})=0$ if and only if $\gamma(v)=0$; 
	\item[(ii)] $\IE(\max[\sum_{\lambda \in \Lambda_v}{{}^\pm|X_\lambda|}; 0])= \frac{1}{2}\partial_t[\Hilbert \{\phi_Y\}](0)$ 
                    if $\gamma(v)=0$.
\end{enumerate}
\label{theorem:CounterpartyRiskDigraph}
\end{theorem}
\begin{proof}
See section \ref{subsec:ProofDigraphTheorem}.%
\end{proof}
A graph is called \textit{Eulerian} if each vertex of that graph has an even degree. A digraph is called \textit{Eulerian} if the in-degree equals the out-degree for each vertex $v$ of that digraph, i.e., if $\gamma(v)=0$ for each vertex $v\in V$.\\

We have shown in the last theorem that a vertex $v\in V$ of a digraph with $\gamma(v)=0$ is distinguished in the context of counterparty credit risk. Because of the definition of an Eulerian digraph $D_k=(V, A^k)$ the equation $\gamma(v)=0$ is valid for every vertex $v\in V$. That is, netting efficiency goes hand in hand with Eulerian digraphs in the context of so-called multilateral netting rules\footnote{Please refer to section \ref{sec:MultilateralNetting} and especially to Example \ref{example:LaplaceMultiNetting}.}. For graphs we can state a similar result.
\begin{theorem}
\label{theorem:graphs}
Let $G_k=(V, E^k)$ be a connected graph with $\mathX \sim P$ and a netting set $\Lambda_v$ with $v\in V$. Then the following holds: 
\begin{align}
    \label{form:Graph}
    \IE\left(\max\left[\sum_{\lambda \in \Lambda_v}{X_\lambda}; 0\right]\right) = \frac{1}{2} \partial_t[\Hilbert\{\phi_Y(t)\}](0).
\end{align}
\label{theorem:CounterpartyRiskGraph}
\end{theorem}
\begin{proof}
See section \ref{subsec:ProofGraphTheorem}.%
\end{proof}
In contrast to digraphs, formula (\ref{form:Graph}) of Theorem \ref{theorem:CounterpartyRiskGraph} is valid for any vertex of an arbitrary graph.\  
The reason for this mismatch is that the symmetry of the net r.v. $Y=\sum_{\lambda \in \Lambda_v}{X_\lambda}$ does not depend on the netting set $\Lambda_v$. In the case of a graph, the r.v.s $X_\lambda$ with $\lambda \in \Lambda_v$ are symmetric and so is $Y$.\\

Thus, using either a graph or a digraph to model a financial market does matter and should be well-considered.

\section{Application of the Network Model}
\label{sec:Application}
In this section we define and explain measures for counterparty credit risk within an OTC and a centrally cleared market. Afterwards, we are going to derive how to compute the expected counterparty credit risk for both types of markets for a typical day in the future by applying the model introduced in section \ref{sec:Model}. For that purpose we clarify, in a first step, the outline of both netting types. In a second step, we apply the notation  $X_{v, w}^{(k)}$ for a r.v. and $x_{v, w}^{(k)} \in \IR$ for its realisation as introduced in section \ref{subsec:MarketSettings}.\\

To this end, we need to define the terms adjacent and neighbourhood. Two different vertices $v$ and $w$ of a graph $G_k=(V, E^k)$ are \textit{adjacent} if $\{v, w\}$ is an edge of $G_k$. In the case of a digraph $D_k=(V, A^k)$, the two vertices are \textit{adjacent} if either $(v,w)\in A^k$ or $(w,v)\in A^k$ is valid. The \textit{neighbourhood} $U\subseteq V$ of a vertex $v$ in a graph $G_k$ or a digraph $D_k$ is the set of all vertices adjacent to $v$.

\subsection{Counterparty Risk within an OTC Market}
\label{sec:BilateralNetting}
Within an OTC market $\calM$ we are allowed to offset positions across all kinds of derivatives classes, but only between one single pair of counterparties. Thus, a bilateral netting set of a given market participant $v\in V$ will correspond to a \textit{bilateral portfolio} of $v$ and one of its counterparts $w\in U_{v}^{C}$. Here, $U_{v}^{C}$ is the neighbourhood of $v$ across \textit{all} classes of derivatives $k\in C$.\\

The real number $x_{v, w}^{(k)}$ shall now represent the current observable position value of $v$ relative to $w\in V\setminus \{v\}$ within derivatives class $k$. The deterministic function value $y_{v,w}(C) := \ \sum_{k\in C}{\ x^{(k)}_{v, w}}$ is called the \textit{current bilateral position of} $v$ \textit{to} $w$. Obviously, the function $y_{v, w}$ can be positive or negative, and an immediate consequence of the definition is the validity of $y_{v,w}(C) = -y_{w, v}(C)$, meaning that the claims of the one are the liabilities of the other counterparty. Thus, the actual \textit{current bilateral counterparty risk of} $v$ \textit{to} $w$ is $\max[y_{v, w}(C); 0]$ and we call $z_\mathb(D):= \sum_{v\in V}{\sum_{w\in U_{v}^{C}}{\max[y_{v, w}(C); 0]}}$ the \textit{current bilateral counterparty risk of} $D=(V, \widehat{A})$. This is a reasonable bilateral counterparty risk measure of an arbitrary market $\calM$ because it adds up the netted exposure of all bilateral portfolios of $\calM$.\\

We now assume that the trade position that corresponds with an arrow of the digraph $D=(V, \widehat{A})$ is not yet realised but represented abstractly by a r.v., i.e., $\mathX \sim {}^\pm|P|$. For all other relations between two different counterparts, we set the (future) position value deterministically to zero. The direction of each arrow in each class $k\in C$ determines whether the positive or the negative absolute value of the associated r.v. is relevant for the calculation.\\ 

We shall call $\IE\left(\max[Y_{v, w}(C);0]\right)$ the \textit{expected bilateral counterparty risk of} $v$ \textit{to} $w$, where $Y_{v, w}(C)=\sum_{k\in C}{{}^\pm|X_{v, w}^{(k)}|}$ is now a conditional r.v. with mean and variance depending on the information about the directions of the positions. In section \ref{subsec:NettingExposure} we have demonstrated how to compute the expectation of such a random variable. Let $\phi_{Y_{v, w}}$ and $\phi_{\max[Y_{v, w};0]}$ be the c.f.s of the r.v. $Y_{v, w}(C)$ and $\max[Y_{v, w}(C);0]$, respectively. If we apply formulas (\ref{form:SumRV}) to (\ref{form:CharFuncMoments}) we obtain
\begin{align}
     \label{form:ExptBilCPR}
	 \IE\left(\max[Y_{v, w}(C);0]\right) = \frac{\partial_t[\phi_{\max[Y_{v, w};0]}(t)](0)}{\ii}.
\end{align}
Putting the parts together and considering the possible asymmetry of the r.v.s of a digraph, as described in section \ref{sec:StructureTheorem}, we finally obtain the formula
\begin{align}
	\label{form:ExptBilCPRMarket}
	\IE(Z_\mathb(D)) = \sum_{ v\in V}{ \sum_{w\in U_{v}^{C}}{\IE\left(\max[Y_{v, w}(C);0] \right)} } 
                    = \sum_{ v\in V}{ \sum_{w\in U_{v}^{C}}{\frac{\partial_t[\phi_{\max[Y_{v, w};0]}(t)](0)}{\ii} } }
\end{align}
for the \textit{expected bilateral counterparty risk of} an arbitrary digraph $D$.
\begin{example}
\label{exa:BilTwoTierStructure}
Financial markets can be organised in different layers. In Germany, for instance, \textsc{Upper} and \textsc{Worms} \cite{Upper2004} describe a two-tier structure of the German interbank market. The directed two-tier structure shown in Fig. \ref{abb:CounterExampleAdvantageousness} with $\mathX \sim {}^\pm|\calL(0,1)|$ is a digraph with $N=6$ market participants and two classes of derivatives $C=\{1,2\}$, where the different classes are depicted by different looking arrows.
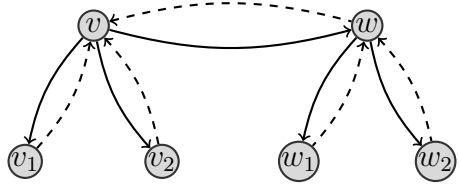
\begin{figure}[htb]
\centering
\tikzstyle{blueStyle}=[->, circle,draw=black!80,fill=gray!30,thick,inner sep=0pt,minimum size=4mm]
\tikzstyle{pre}=[<-,shorten <=1pt,>=stealth',semithick]
\tikzstyle{post}=[->,shorten >=1pt,>=stealth',semithick]
\tikzstyle{every loop}=[min distance=15mm, looseness=1.5,out=30, in=-30]
\begin{tikzpicture}[scale=0.9]
\node[blueStyle] (v) at (-2,0) {$v$};
\node[blueStyle] (v_1) at (-3,-2) {$v_1$};
\node[blueStyle] (v_2) at (-1,-2) {$v_2$};
\node[blueStyle] (w) at (2,0) {$w$};
\node[blueStyle] (w_1) at (1,-2) {$w_1$};
\node[blueStyle] (w_2) at (3,-2) {$w_2$};
\draw[->, black!100, thick] (v) to[bend left=-15] node[above, font=\small] {} (w);
\draw[->, black!100, thick] (v) to[bend left=-15] node[below, font=\small] {} (v_1);
\draw[->, black!100, thick] (v) to[bend left=-15] node[above, font=\small] {} (v_2);
\draw[->, black!100, thick] (w) to[bend left=-15] node[below, font=\small] {} (w_1);
\draw[->, black!100, thick] (w) to[bend left=-15] node[above, font=\small] {} (w_2);
\draw[<-, black!100, thick, dashed] (v) to[bend right=-15] node[above, font=\small] {} (w);
\draw[<-, black!100, thick, dashed] (v) to[bend right=-15] node[below, font=\small] {} (v_1);
\draw[<-, black!100, thick, dashed] (v) to[bend right=-15] node[above, font=\small] {} (v_2);
\draw[<-, black!100, thick, dashed] (w) to[bend right=-15] node[below, font=\small] {} (w_1);
\draw[<-, black!100, thick, dashed] (w) to[bend right=-15] node[above, font=\small] {} (w_2);
\end{tikzpicture}
\caption{Digraph with a directed two-tier structure and two classes of derivatives.}
\label{abb:CounterExampleAdvantageousness}
\end{figure}
Calculating the expected bilateral counterparty risk for the digraph $D=(V, \widehat{A})$ with $\widehat{A}= A^1 \cupdot A^2$ as depicted in Fig. \ref{abb:CounterExampleAdvantageousness} is straightforward, because all bilateral portfolios have the same simple structure. There is only one claim and one debt within each of the ten bilateral portfolios. For instance, $Y_{v, w}(C) = |X_{v, w}^{(1)}| - |X_{v, w}^{(2)}|$, whereby the r.v.s on the right hand side of the last equation are distributed by the positve and negative absolute value of the Laplace distribution. The associated c.f. $\phi_Y$ of $Y_{v, w}(C)$ equals $\frac{1}{1+t^2}=\frac{\ii}{\ii+t}\frac{-\ii}{-\ii+t}$ and applying formulas (\ref{form:MaxOfRV}) to (\ref{form:CharFuncMoments}) we get $\phi_{\max[Y; 0]}(t) \ = \ \frac{1}{2}[1+\frac{1}{1+t^2}] +  \frac{\ii}{2}[\Hilbert\{\frac{1}{1+t^2}\}(t) - \Hilbert\{\frac{1}{1+t^2}\}(0)]= \frac{2\ii +t}{2\ii +2t}$, where $\Hilbert\{\frac{1}{1+t^2}\}(t) = \frac{t}{1+t^2}$.%
\footnote{Please note that the positive absolute value of the Laplace distribution is the exponential distribution. The imaginary part of the c.f. of the exponential distribution equals the Hilbert transform of $\frac{1}{1+t^2}$ because of equation (\ref{form:PosPartAnalyticSignal}). An alternative way to derive the Hilbert transform of $\frac{1}{1+t^2}$ is to use formula (\ref{form:ResidueFormula}) or (\ref{form:SimplyfiedResidueFormula}). In our case we get $\Hilbert\{ \frac{1}{1+t^2} \}(\omega) = 2\ii \Res\left[\frac{1}{(1+t^2)(\omega-t)},\ii\right] + \ii \Res\left[\frac{1}{(1+t^2)(\omega-t)}, \omega\right] =  \frac{\ii}{-\omega^2-1} + \frac{1}{\omega-\ii} = \frac{\omega}{1+\omega^2}$ and by turning $\omega$ to $t$ we get the same result. Please refer to Example \ref{example:LaplaceMultiNetting}. } %
Thus, $\IE(\max[Y_{v, w}(C);0])= \frac{ \partial_t[\frac{2\ii +t}{2\ii +2t}](0)}{\ii}=\frac{1}{2}$ is the expected counterparty risk for $v$ relative to $w$. Because of formula (\ref{form:ExptBilCPRMarket}) and the similarity of the bilateral portfolios, we obtain $\IE(Z_\mathb(D)) = \frac{10}{2}=5$.
\end{example}
Obviously, if we leave out the information about the direction of the r.v.s and if we further consider the symmetry of the r.v.s of a graph we will receive a similar formula for the underlying graph $G_D$ of $D$. The next example will demonstrate this obvious result. 
\begin{example}
\label{exa:BilCPR:ComplGraph}
Consider the graph $G:=(\{v, w\}, \widehat{E})$ with $N=2$, $\mathX \sim \calN(0, \sigma)$ and $\widehat{E}= E^1 \cupdot \ldots \cupdot E^K$ as sketched in Fig \ref{abb:BilatCounterpartyRisk}, where the different looking edges represent the netted positions between $v$ and $w$ within one of $K$ derivatives classes.\\

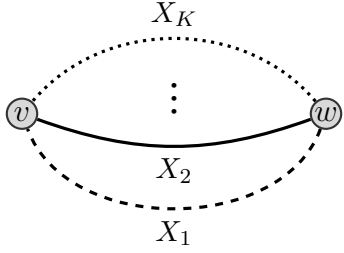
\begin{figure}[htb]
\centering
\tikzstyle{blueStyle}=[->, circle,draw=black!80,fill=gray!30,thick,inner sep=0pt,minimum size=4mm]
\tikzstyle{transparentStyle}=[->, circle,draw=white!100,fill=white!100]
\tikzstyle{pre}=[<-,shorten <=1pt,>=stealth',semithick]
\tikzstyle{post}=[->,shorten >=1pt,>=stealth',semithick]
\tikzstyle{every loop}=[min distance=15mm, looseness=1.5,out=30, in=-30]

\begin{tikzpicture}
\node[blueStyle] (v1) at (0,0) {$v$};
\node[blueStyle] (v2) at (4,0) {$w$};
\node[transparentStyle] (d1) at (2,0.2) {\Large$\vdots$};
\draw[-, black!100, very thick, dashed] (v2) to[bend right=-70] node[below, font=\small] {$X_1$} (v1);
\draw[-, black!100, very thick] (v1) to[bend left=-20] node[below, font=\small] {$X_2$} (v2);
\draw[-, black!100, very thick, dotted] (v2) to[bend left=-50] node[above, font=\small] {$X_K$} (v1);
\end{tikzpicture}
\caption{Two-vertex graph $G$ with $K$ edges}
\label{abb:BilatCounterpartyRisk}
\end{figure}

We have assumed $X_{v, w}^{(k)}\sim \calN(0,\sigma)$ for all $k\in C$ and we do not have any additional information about the direction of the edges. The c.f. $\phi_{Y_{v, w}(C)}(t)=\phi_{Y_{w,v}(C)}(t)=e^{-\frac{t^2 K \sigma^2}{2}}$ of the sum $Y:=Y_{v, w}(C)=\sum_{k\in C}{X_{v, w}^{(k)}}$ equals the \textit{even} and real-valued c.f. of $\calN(0,\sqrt{K}\sigma)$. Because of equation (\ref{form:MaxOfRV}) and $\Hilbert\{ \phi_{Y} \}(0)=0$ we obtain
\begin{align*}
	\phi_{\max[Y;0]}(t) 	
                    &=  \frac{1}{2}[1+\phi_{Y}(t)] +  \frac{\ii}{2}[\Hilbert\{ \phi_{Y} \}(t)-\Hilbert\{ \phi_{Y} \}(0)] \\
				&= \frac{1}{2}+\frac{e^{\frac{-t^2 K\sigma^2}{2}}}{2}+ \ii \frac{\Dawson\left( \frac{t \sqrt{K}\sigma}{\sqrt{2}}\right)}{\sqrt{\pi}}
\end{align*}
where $\Hilbert\{ \phi_{Y} \}(t) = \frac{2}{\sqrt{\pi}}\Dawson\left( \frac{t \sqrt{K}\sigma}{\sqrt{2}}\right)$ and $\Dawson(t):= e^{-t^2} \int_{0}^{t}{e^{s^2} \, ds}$ is the so-called \textit{Dawson function}\footnote{See 7.2.5 in \cite{Olver2013}.}. The expected bilateral counterparty risk of $v$ to $w$ is then given by
\begin{align*}
	\IE(\max[Y;0]) \ = \  \frac{\partial_t[\phi_{\max[Y;0]}](0)}{\ii} \ = \  \frac{\partial_t\left[\ii \frac{\Dawson\left( \frac{t \sqrt{K}\sigma}{\sqrt{2}}\right)}{\sqrt{\pi}}\right](0)}{\ii} \ = \ \sigma\sqrt{\frac{K}{2\pi}}.
\end{align*}
Please also refer to formula (\ref{form:Graph}). The expectation above can be generalised to a complete graph with $N>2$ market participants. Each market participant $v$ would have $(N-1)$ counterparties in each derivatives class $k\in C$. The expected counterparty risk of one market participant would then be $(N-1)\sigma\sqrt{\frac{K}{2\pi}}$, which matches formula (3) in \cite{Duffie2011}. Please note the symmetry of the bilateral portfolios and that this expectation comprises $2^{\frac{N(N-1)}{2}}$ orientations for each of the $K$ derivatives classes. 
\end{example}
\begin{concl} 
Let $D=(V, \widehat{A})$ be a connected digraph along with $N\geq 2$ market participants, $K>1$ derivatives classes $C=\{1, \ldots, K\}$ and $\mathX \sim {}^\pm|P|$ that represents an OTC market. Then the following holds:
\begin{enumerate}
	\item[(i)]  $\IE[\sum_{k\in C}{{}^\pm|X_{v, w}^{(k)}|}] = \IE[Y_{v, w}(C)]=0$ if and only if the subdigraph induced by $\{v, w\}$ is Eulerian;
	\item[(ii)] $\IE\left(\max[Y_{v, w}(C);0]\right) = \frac{1}{2} \partial_t [\Hilbert \{\phi_{Y_{v, k}}(t)](0)$ if the 
                    subdigraph induced by $\{v, w\}$ is Eulerian.
\end{enumerate}
\label{concl:EulerianBilatCCR}
\end{concl}
The proof of this conclusion is obvious since the netting set of a bilateral market equals all arrows between $v$ and $w$. Then, we only have to apply Theorem \ref{theorem:CounterpartyRiskDigraph}.

\subsection{Counterparty Risk within a Centrally Cleared Market}
\label{sec:MultilateralNetting}
By introducing a CCP we do not need to extend the preceding setting, because the network contains enough information to calculate the counterparty credit risk of each single entity and therefore of the entire network. That is, we imagine a CCP as an abstract entity which is not part of the market participant's network. We assume that a CCP clears exactly one class of derivatives. Within a multilateral derivatives market $\calM$ we are not allowed to offset positions across different derivatives classes. But we can aggregate positions of all market participants of one single class $k\in C$ that is cleared by a CCP instead. Thus, a multilateral netting set of a market participant $v\in V$ equals the neighbourhood $U_{v}^{(k)}$ of $v$ within derivatives class $k\in C$.\\

We shall call $y_{v,k}(U_{v}^{(k)}) :=  \sum_{w\in U_{v}^{(k)}}{\ x_{v, w}^{(k)}} \in \IR$ the \textit{current multilateral position of} $v$ within derivatives class $k$. The function $y_{v,k}$ therefore represents the netted sum of the observable trade positions of $v$ relative to all adjacent 
vertices within class $k$, but the actual \textit{current multilateral counterparty risk of} $v$ to the CCP is $\max[y_{v, k}(U_{v}^{(k)}); 0]$.\\ 

If we change the perspective to the overall counterparty risk of the entire market, we get $z_{\mathm}(D_k) := \sum_{v\in V}{\left|y_{v, k}(U_{v}^{(k)})\right|}$ the \textit{current multilateral counterparty risk of} $D_k$. The equation $\sum_{v\in V}{y_{v, k}(U_{v}^{(k)})} = 0$ is valid since each position is counted \textsl{twice} in $z_\mathm(D_k)$, once as debt and once as claim. Therefore, we receive
\begin{align*}
	z_\mathm(D_k) = \sum_{v\in V}{\left|y_{v, k}(U_{v}^{(k)})\right|} =	2\cdot\sum_{v\in V}{\max[y_{v, k}(U_{v}^{(k)}); 0]}.
\end{align*} 
The last equation justifies the novation of the original bilateral contract to a CCP. We further assume that bilateral netting is the prevailing form of netting, unless explicitly stated otherwise. Therefore, if we introduce a CCP in one of the $K$ classes of derivatives, let us say in the class $k\in C$, then we assume that the remaining $(K-1)$ classes are still cleared bilaterally. Thus the total exposure of the entire market is 
\begin{align}
     \label{form:MultilExpEntireDigraph}
		z_\mathm(D) := \ z_\mathm(D_k) + z_\mathb(D \setminus D_k),
\end{align}
where the last summand of the right-hand side of the equation denotes the current bilateral counterparty risk of the OTC market $D \setminus D_k :=(V, \widehat{A} \setminus A^k)$. \\

Let us now suppose that the position values of $\calM$ that corresponds to arrows of the digraph $D = (V,A)$ are not yet realised but represented abstractly by $\mathX \sim {}^\pm|P|$. For all other relations between two different counterparts, we set the (future) position value deterministically to zero. The direction of each arrow of $A^k$ determines the condition $\pm$ of the r.v.s ${}^\pm|X_{v, w}^{(k)}|$ distributed by ${}^\pm|P|$. Further, we denote the c.f. of the r.v. $Y_{v, k}(U_{v}^{(k)})$ and $\max[Y_{v, k}(U_{v}^{(k)});0]$ by $\phi_{Y_v}$ and $\phi_{\max[Y_v;0]}$, respectively. If we apply again formulas (\ref{form:SumRV}) to (\ref{form:CharFuncMoments}) the \textit{expected multilateral counterparty risk of} of an arbitrary market participant $v$ of $D_k$ is then determined by 
\begin{align}
     \label{form:MultCCR_Pair}
     \IE\left(\max[Y_{v, k}(U_{v}^{(k)});0]\right) \ = \ \frac{\partial_t[\phi_{\max[Y_v;0]}(t)](0)}{\ii}, 
\end{align}
where $Y_{v,k}(U_{v}^{(k)})=\sum_{w\in U_{v}^{(k)}}{{}^\pm|X_{v, w}^{(k)}|}$ is a r.v. with mean and variance depending on the information about the directions of the positions. Putting the parts together, we obtain the formula
\begin{align}
	\label{form:MultCCR_Market}
	\IE(Z_\mathm(D_k)) = \sum_{v\in V}{ \IE\left(\max[Y_{v, k}(U_{v}^{(k)});0]\right) } = \sum_{v\in V}{ \frac{\partial_t[\phi_{\max[Y_v;0]}(t)](0)}{\ii} }
\end{align} 
for the \textit{expected multilateral counterparty risk of} an arbitrary digraph $D_k$ (of derivatives class $k$). We adopt formula (\ref{form:MultilExpEntireDigraph}) to get 
\begin{align}
     \label{form:CCP:OverallExposure:Digraph}
     Z_\mathm(D) := \ Z_\mathm(D_k) + Z_\mathb(D \setminus D_k)
\end{align}
for digraphs, where the weights of the arrows depend on random variables.\\

Again, if we fade the information about the direction of each r.v. out, we will obtain similar formulas like (\ref{form:MultCCR_Pair}) and (\ref{form:MultCCR_Market}) for an undirected underlying graph $G_k$ of $D_k$. As in Example \ref{exa:BilCPR:ComplGraph}, we derive a formula for a complete undirected graph with $\mathX \sim \calN(0,\sigma)$.
\begin{example}
Suppose $G=(V, E)$ is the undirected complete graph with $N$ vertices, one class of derivatives and $\mathX \sim \calN(0,\sigma)$. We do not have any information about the direction of the positions. For the sake of simplicity we write $Y_{v}$ instead of $Y_{v,1}=\sum_{w\in U_{v}^{(k)}}{X_{v, w}^{(1)}}$ and $X_{v, w}$ instead of $X_{v, w}^{(1)}$. The c.f. of the sum $Y_{v}(U_{v}^{(k)})=\sum_{w\in U_{v}^{(k)}}{X_{v, w}}$ equals the \textit{even} function $\phi_{Y_{v}}(t) = e^{ \frac{-t^2}{2} (N-1) \sigma^2}$ and its Hilbert transform equals the \textit{odd} function $\Hilbert \{ \phi_{Y_{v}}\}(t) = \frac{2}{\sqrt{\pi}}\Dawson\left(\frac{t\sqrt{N-1}\sigma}{\sqrt{2}}\right)$. Applying formulas (\ref{form:Graph}) and (\ref{form:MultCCR_Pair}), we obtain
\begin{align*}
	\IE\left(\max[Y_{v}(U_{v}^{(k)});0]\right) \ = \ \frac{1}{\ii}\partial_t\left[\frac{1}{2}\left(1+e^{\frac{-t^2}{2}(N-1)\sigma^2}\right)+ \ii\frac{\Dawson\left(\frac{t\sqrt{N-1}\sigma}{\sqrt{2}}\right)}{\sqrt{\pi}}\right](0) = \sqrt{\frac{N-1}{2\pi}}\sigma
\end{align*}
the expected multilateral counterparty risk of $v$ within a complete graph with $N$ vertices. The last equation equals formula (4) in \cite{Duffie2011}. Be aware of the symmetry of the multilateral portfolios and that $G$ is the underlying graph of $2^{\frac{N(N-1)}{2}}$ different digraphs and the calculated expected counterparty risk is some sort of average of all of them.
\end{example}
Applying Theorem \ref{theorem:CounterpartyRiskDigraph} to a centrally cleared market we get the following conclusion.  
\begin{concl} 
Let $D=(V, \widehat{A})$ be a connected digraph along with $N\geq 2$ market participants, $K$ derivatives classes $C=\{1, \ldots, K\}$ and $\mathX \sim {}^\pm|P|$ that represents a centrally cleared market. Then the following holds:
\begin{enumerate}
	\item[(i)]  $\IE[\sum_{w\in U_{v}^{(k)}}{{}^\pm|X_{v, w}^{(k)}|}] =\IE[Y_{v, k}(U_{v}^{k})]=0$ if and only if $\gamma(v)=0$ within class $k$;
	\item[(ii)] $\IE\left(\max[Y_{v, k}(U_{v}^{(k)});0]\right) = \frac{1}{2} \partial_t [\Hilbert \{\phi_{Y_{v, k}}(t)](0)$ if $\gamma(v)=0$ 
                    within class $k$.
\end{enumerate}
\label{concl:EulerianMultiCCR}
\end{concl}
\begin{example}
\label{example:LaplaceMultiNetting}
Suppose $G=(V, E=\{e_1, e_2, e_3\})$ is the undirected graph as depicted in Fig. \ref{abb:C3} with a single class of derivatives and $\mathX \sim \calL(0, 1)$. Let $l, m, n \in \{1,2,3\}$ with $m \neq n$.\\

\begin{figure}[htb]
\centering
\tikzstyle{blueStyle}=[->, circle,draw=black!80,fill=gray!30,thick,inner sep=0pt,minimum size=4mm]
\tikzstyle{pre}=[<-,shorten <=1pt,>=stealth',semithick]
\tikzstyle{post}=[->,shorten >=1pt,>=stealth',semithick]
\tikzstyle{every loop}=[min distance=15mm, looseness=1.5,out=30, in=-30]
\begin{tikzpicture}[scale=0.7]
\node[blueStyle] (v_1) at (0,0) {$v_1$};
\node[blueStyle] (v_2) at (2,0) {$v_2$};
\node[blueStyle] (v_3) at (0,2) {$v_3$};
\draw[-, black!100, thick] (v_1) to node[below, font=\small] {$e_1$} (v_2);
\draw[-, black!100, thick] (v_2) to node[above, font=\small] {$e_2$} (v_3);
\draw[-, black!100, thick] (v_3) to node[left, font=\small] {$e_3$} (v_1);
\end{tikzpicture}
\hspace{2cm}
\begin{tikzpicture}[scale=0.7]
\node[blueStyle] (v_1) at (0,0) {$v_1$};
\node[blueStyle] (v_2) at (2,0) {$v_2$};
\node[blueStyle] (v_3) at (0,2) {$v_3$};
\draw[->, black!100, thick] (v_1) to node[below, font=\small] {$a_1$} (v_2);
\draw[->, black!100, thick] (v_2) to node[above, font=\small] {$a_2$} (v_3);
\draw[->, black!100, thick] (v_3) to node[left, font=\small] {$a_3$} (v_1);
\end{tikzpicture}
\caption{Graph $G$ and an Eulerian orientation $D$}
\label{abb:C3}
\end{figure}
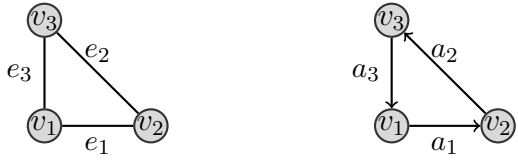
We consider the Eulerian orientation $D=(V, \{a_1, a_2, a_3\})$ as sketched in Fig. \ref{abb:C3}, which is one out of eight possible orientations of $G$ as depicted in Fig. \ref{abb:Orientations}. The orientations define the directions of the edges and therefore the direction of the corresponding trade positions. 
\begin{figure}[htb]
\centering
\tikzstyle{blueStyle}=[->, circle,draw=black!80,fill=gray!30,thick,inner sep=0pt,minimum size=4mm]
\begin{tikzpicture}[scale=0.65]
\node[blueStyle, font=\small] (v1) at (0,0) {$v_1$};
\node[blueStyle, font=\small] (v_2) at (2,0) {$v_2$};
\node[blueStyle, font=\small] (v_3) at (0,2) {$v_3$};
\draw[->, black!100, thick] (v_2) to node{} (v_1);
\draw[->, black!100, thick] (v_3) to node{}  (v_2);
\draw[->, black!100, thick] (v_1) to node{}  (v_3);
\end{tikzpicture}
\hspace{1cm}
\begin{tikzpicture}[scale=0.65]
\node[blueStyle, font=\small] (v_1) at (0,0) {$v_1$};
\node[blueStyle, font=\small] (v_2) at (2,0) {$v_2$};
\node[blueStyle, font=\small] (v_3) at (0,2) {$v_3$};
\draw[->, black!100, thick] (v_1) to node{} (v_2);
\draw[->, black!100, thick] (v_3) to node{}  (v_2);
\draw[->, black!100, thick] (v_1) to node{}  (v_3);
\end{tikzpicture}
\hspace{1cm}
\begin{tikzpicture}[scale=0.65]
\node[blueStyle, font=\small] (v_1) at (0,0) {$v_1$};
\node[blueStyle, font=\small] (v_2) at (2,0) {$v_2$};
\node[blueStyle, font=\small] (v_3) at (0,2) {$v_3$};
\draw[->, black!100, thick] (v_2) to node{}  (v_1);
\draw[->, black!100, thick] (v_2) to node{}  (v_3);
\draw[->, black!100, thick] (v_1) to node{}  (v_3);
\end{tikzpicture}
\hspace{1cm}
\begin{tikzpicture}[scale=0.65]
\node[blueStyle, font=\small] (v_1) at (0,0) {$v_1$};
\node[blueStyle, font=\small] (v_2) at (2,0) {$v_2$};
\node[blueStyle, font=\small] (v_3) at (0,2) {$v_3$};
\draw[->, black!100, thick] (v_2) to node{} (v_1);
\draw[->, black!100, thick] (v_3) to node{}  (v_2);
\draw[->, black!100, thick] (v_3) to node{}  (v_1);
\end{tikzpicture}
\newline
\newline
\begin{tikzpicture}[scale=0.65]
\node[blueStyle, font=\small] (v_1) at (0,0) {$v_1$};
\node[blueStyle, font=\small] (v_2) at (2,0) {$v_2$};
\node[blueStyle, font=\small] (v_3) at (0,2) {$v_3$};
\draw[->, black!100, thick] (v_1) to node{} (v_2);
\draw[->, black!100, thick] (v_2) to node{}  (v_3);
\draw[->, black!100, thick] (v_3) to node{}  (v_1);
\end{tikzpicture}
\hspace{1cm}
\begin{tikzpicture}[scale=0.65]
\node[blueStyle, font=\small] (v_1) at (0,0) {$v_1$};
\node[blueStyle, font=\small] (v_2) at (2,0) {$v_2$};
\node[blueStyle, font=\small] (v_3) at (0,2) {$v_3$};
\draw[->, black!100, thick] (v_2) to node{} (v_1);
\draw[->, black!100, thick] (v_2) to node{}  (v_3);
\draw[->, black!100, thick] (v_3) to node{}  (v_1);
\end{tikzpicture}
\hspace{1cm}
\begin{tikzpicture}[scale=0.65]
\node[blueStyle, font=\small] (v_1) at (0,0) {$v_1$};
\node[blueStyle, font=\small] (v_2) at (2,0) {$v_2$};
\node[blueStyle, font=\small] (v_3) at (0,2) {$v_3$};
\draw[->, black!100, thick] (v_1) to node{}  (v_2);
\draw[->, black!100, thick] (v_3) to node{}  (v_2);
\draw[->, black!100, thick] (v_3) to node{}  (v_1);
\end{tikzpicture}
\hspace{1cm}
\begin{tikzpicture}[scale=0.65]
\node[blueStyle, font=\small] (v_1) at (0,0) {$v_1$};
\node[blueStyle, font=\small] (v_2) at (2,0) {$v_2$};
\node[blueStyle, font=\small] (v_3) at (0,2) {$v_3$};
\draw[->, black!100, thick] (v_1) to node{} (v_2);
\draw[->, black!100, thick] (v_2) to node{}  (v_3);
\draw[->, black!100, thick] (v_1) to node{}  (v_3);
\end{tikzpicture}
\hspace{1.3cm}
\caption{All $8=2^3$ possible orientations of $G$}
\label{abb:Orientations}
\end{figure}
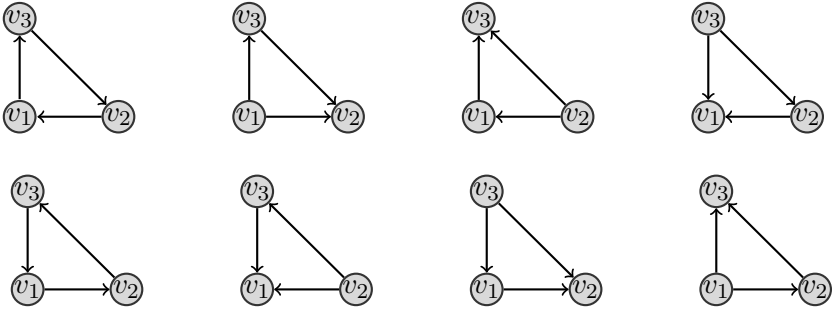
The c.f. $\frac{1}{1+t^2}=\frac{\ii}{\ii +t}\frac{-\ii}{-\ii +t}$ of the netted sum $Y_{v_l}:= |X_m|-|X_n|$ is real-valued and even. Here, $X_m$ and $X_n$ are r.v.s that represent the corresponding positions. Applying (ii) of Conclusion \ref{concl:EulerianMultiCCR} we obtain
\begin{align*}
	\IE\left(\max[Y_{v_l}(U_{v_l}^{(1)}); 0 ]\right) \ = \ \frac{1}{2}\partial_t [\Hilbert \{\phi_{Y_{v_l}}\}(t)](0) 
     \ = \ \frac{1}{2}\partial_t \left[\frac{t}{1+t^2}\right](0) = \frac{1}{2}
\end{align*}
the expected multilateral counterparty risk of $v_l$ with $l \in \{1,2,3\}$ of the Eulerian orientation $D$. The expectation of the centrally cleared market is therefore given by $\IE(Z_\mathm(D))=\IE(Z_\mathm(D_1)) = \sum_{l=1}^{3}{ \IE\left(\max[Y_{v_l}(U_{v_l}^{(1)}); 0 ]\right) } = 3 \cdot \frac{1}{2}= \frac{3}{2}$. All other six non-Eulerian orientations entail $\frac{5}{2}$ as expected counterparty risk of the entire market, which is significantly more than the risk of the two Eulerian orientations.\\

Let us now turn our attention back to the undirected graph $G=(V, E)$. The c.f. of a position equals the real-valued and even function $\frac{1}{1+t^2}$. The positive absolute value of $\frac{1}{1+t^2}$ is an analytic signal, because $\frac{1}{1+t^2}$ as well as its Fourier transform is an absolute integrable function.\footnote{Please refer to section \ref{subsec:PositiveNegativePart}.} Thus, the imaginary part of $\frac{1}{1 + t^2} + \ii \frac{t}{1+t^2}$ is the Hilbert transform of $\frac{1}{1+t^2}$. \\

The c.f. $\phi_{Y_{v_l}}(t) = \frac{1}{(1+t^2)^2} = \frac{1}{1+t^2}\frac{1}{1+t^2}$ of the r.v. $Y_{v_l} := X_m + X_n$ of participant $v_l$ is also  even and real-valued. Its Hilbert transform $\Hilbert\{\phi_{Y_{v_l}}\}(t)$ for all $l\in \{1,2,3\}$ is given by $\frac{t(3+t^2)}{2(1+t^2)^2}$. Applying formula (\ref{form:Graph}) we obtain
\begin{align*}
	\IE\left(\max[Y_{v_l, 1}(U_{v_l}^{(1)}); 0 ]\right) \ = \ \frac{1}{2}\partial_t\left[\frac{t(3+t^2)}{2(1+t^2)^2}\right](0) \ = \  \frac{3}{4}
\end{align*}
for one of the three vertices of $G$, so $\IE(Z_\mathm(G)) = \frac{9}{4}$. This result can also be deduced as an average of all orientations, that is, $\frac{1}{8}(2\cdot \frac{3}{2} + 6 \cdot \frac{5}{2}) = \frac{9}{4}$.    
\end{example}
The last example has shown that information about the direction of positions is essential for the computation of the expectation of an entire market. Disregarding such information could lead to over- or underestimating counterparty credit risk. It is obvious that similar result can be obtained in more complex networks and that the precise structure is essential for the study of counterparty and therefore systemic risk.

\subsection{Advantageousness of Multilateral Netting}
\label{subsec:Advantageousness}
%\footnote{Note that $\partial_t \Hilbert \{ f(t) \}(t) = \Hilbert \{\partial_t f(t)\}(t)$ is valid. See for example \cite{King2009}, Section 4.8}
%
The objective of this section is to generalise a main result of \textsc{Duffie} and \textsc{Zhu} \cite{Duffie2011}. The authors raised the question of whether it is more advantageous for the overall counterparty credit risk to clear via a CCP or classically bilateral between the two involved counterparties. By applying the introduced network model we can answer this question not only for complete graphs, but also for arbitrary graphs and digraphs. Moreover, the network model introduced in section \ref{sec:Model} is not constrained to the normal distribution, it can also employ any (symmetric) distribution with a defined expected value.\\

Introducing a CCP for a single class of derivatives $k\in C$ within a digraph $D$ with $K\in \IN$ classes of derivatives along with $\mathX \sim {}^\pm|P|$ improves the netting efficiency if and only if
\begin{align*}
	\IE(Z_\mathm(D)) \ = \ \IE(Z_\mathm(D_k)) + \IE(Z_\mathb(D \setminus D_k)) \ < \ \IE(Z_\mathb(D))
\end{align*}
due to definition (\ref{form:CCP:OverallExposure:Digraph}). Because of formulas (\ref{form:ExptBilCPRMarket}) and (\ref{form:MultCCR_Market}) this applies if and only if
\begin{align}
     \label{form:digraph:AdvantageousnessOfCCP}
	\sum_{v\in V}{\frac{\partial_t[{}_{N-1}\phi_{v, k}(t)](0)}{\ii} } + \sum_{v\in V}{\sum_{w\in U_{v}^{C}}{ \frac{\partial_t[{}_{K-1}\phi_{v, w}(t)](0)}{\ii} }}  \ < \  
	\sum_{v\in V}{\sum_{w\in U_{v}^{C}}{ \frac{\partial_t[{}_{K}\phi_{v, w}(t)](0)}{\ii} }},
\end{align}
where 
\begin{align*}
	{}_{M}\phi_{v, \cdot}(t) &:= \frac{1}{2}\left[1 + \phi_X(t)^M \right] + \frac{\ii}{2}\left[ \Hilbert \{\phi_X(t)^M\}(t) - \Hilbert\{\phi_X(t)^M\}(0) \right]
\end{align*}
is the c.f. of the r.v. $\max[Y, 0]$ with $Y := \sum_{j=1}^{M}{{}^\pm|X_j|}$ and $X_j \sim P$. In this section $M\in \IN$ is a representative for $N-1$, $K$ or $K-1$ and we write $X$ instead of $X_j$. We are now able to compute the advantageousness for a specific digraph by applying the introduced stochastic framework to the inequality (\ref{form:digraph:AdvantageousnessOfCCP}).\\ 

Let us now focus on a single \textit{representative counterpart} $v$ of a \textit{complete} undirected \textit{graph} with $\mathX \sim P$ as in \cite{Duffie2011}. Because of the features of a complete graph as well as the equations (\ref{form:ExptBilCPR}) and (\ref{form:MultCCR_Pair}) it is then profitable for $v$ to be cleared via a CCP if and only if
\begin{align}
		\label{form:InequAdvCCP}
	  \frac{\partial_t[{}_{N-1}\phi_{v, k}(t)](0)}{\ii} + (N-1) \frac{\partial_t[{}_{K-1}\phi_{v, w}(t)](0)}{\ii} \ &< \ (N-1)\frac{\partial_t[{}_{K}\phi_{v, w}(t)](0)}{\ii}\\
	  \Leftrightarrow	 (N-1)(\partial_t[{}_{K}\phi_{v, w}(t)](0) - \partial_t[{}_{K-1}\phi_{v, w}(t)](0))   \ &> \   \partial_t[{}_{N-1}\phi_{v,K}(t)](0) .
\end{align}

If we apply formula (\ref{form:Graph}) and then put the result into (\ref{form:InequAdvCCP}) we will receive the inequality
\begin{align}
	\label{form:UndirectedInequ}
     \frac{1}{2} \partial_t[\Hilbert\{\phi_X(t)^{N-1}\}](0)  \ < \  \frac{(N-1)}{2} \left[\partial_t[\Hilbert\{\phi_X(t)^{K}\}](0) - \partial_t[\Hilbert\{\phi_X(t)^{K-1}\}](0) \right]
\end{align}
for a complete graph where the positions are distributed by $P$.

\begin{example}
\label{exa:Advant:Normal}
Suppose we have a complete graph $G=(V, E)$ with $N$ market participants, a single class of derivatives and along with $\mathX \sim \calN(0,1)$. Applying formula (\ref{form:UndirectedInequ}) and considering $\partial_t \Hilbert\{\phi_X(t)^M\}(0) = \sqrt{M}\sqrt{\frac{2}{\pi}}$ we obtain
\begin{align*}
  \frac{1}{2}\sqrt{N-1}\sqrt{\frac{2}{\pi}} \ &< \  \frac{(N-1)}{2} \left( \sqrt{K}\sqrt{\frac{2}{\pi}} - \sqrt{K-1}\sqrt{\frac{2}{\pi}} \right). 
\end{align*}
This inequality can be easily transformed to $K< \frac{N^2}{4(N-1)}$ with $N>2$, which equals formula (6) in \cite{Duffie2011}.
\end{example}

The Laplace distribution has plainly fatter tails than the normal distribution. However, the results of the next example are very similar compared to these of Example \ref{exa:Advant:Normal}. 

\begin{example}
\label{exa:AdvantageousnessCCP}
Again, we consider a complete graph $G=(V, E)$ along with $\mathX \sim \calL(0,1)$. The c.f. $\phi_X(t)$ equals $\frac{1}{1+t^2}$. Unfortunately, in this case we can not solve the inequality induced by formula (\ref{form:UndirectedInequ}) exactly. Instead, we apply formula (\ref{form:Graph}) in order to obtain
\begin{align}
	\label{form:ExpCRiskLaplace}
	 \frac{1}{\ii} \partial_t [ {}_{M}\phi_{v, \cdot}(t) ](0) = \frac{1}{2} \partial_t [ \Hilbert\{\phi_{X}(t)^M \} ](0) = \frac{\Gamma(\frac{1}{2}+M)}{\sqrt{\pi}\Gamma(M)} = \frac{M}{2^{2M}} \binom{2M}{M}
\end{align}
the expected multilateral counterparty risk for the representative counterpart $v$ within derivatives class $k$. Here, $\Gamma(t) := \int_0^\infty{x^{t-1} e^{-x} \,dx}$ is the so-called \textit{gamma function}, which is an extension of the factorial function.\footnote{The result of Example \ref{example:LaplaceMultiNetting} can now be recalculated with formula (\ref{form:ExpCRiskLaplace}). Just set $M=N-1=2$ and we get $\frac{1}{\ii} \partial_t [ {}_{2}\phi_{v, k}(t)] = \frac{\frac{3\sqrt{\pi}}{4}}{\sqrt{\pi}} = \frac{3}{4}$.} % 
\begin{table}[htb]
\centering
\begin{tabular}{c|cccccccccc} 
\hline
$K=$	     & 1 & 2 & 3  & 4  & 5  & 6  & 7  & 8  & 9  & 10 \\
\hline
$N\geq$   & 2 & 6 & 10 & 14 & 18 & 22 & 26 & 30 & 34 & 38 \\
\hline
\end{tabular}
	\caption{How many market participants do we need to assure the advantageousness of the central clearing?}
		\label{tab:Laplace}
\end{table}
By applying formula (\ref{form:InequAdvCCP}) and (\ref{form:ExpCRiskLaplace}) we can deduce the values of Tab. \ref{tab:Laplace}, where the solution set of $N$ for a given $K$ is listed. For instance, if we have $K=5$ classes of derivatives, we will need at least $N=18$ market participants in order to ensure the advantageousness of the central clearing under the presuppositions made.
\end{example}
We have confirmed formula (6) of \cite{Duffie2011}, and we have also shown how to apply the introduced model to do a similar calculation for a different distribution within a complete graph. However, if we want to study the exact constitution of a financial market we need to study an arbitrary graph or digraph and not a specific type of graphs or digraphs. As a result something as a single representative counterpart for the entire structure can in general not exist. However, we can compare the implications made from the different models in order to validate the associate assumptions.
\begin{example}
Let $G_D$ be the underlying graph $G_D$ of Example \ref{exa:BilTwoTierStructure}. That is, we consider an undirected two-tier market structure with $\mathX \sim \calL(0,1)$. On the one hand, the expected bilateral counterparty risk $\IE(Z_\mathb(G_D))=7.5$ is significantly smaller than the expected multilateral counterparty risk $\IE(Z_\mathm(G_D)) = 8.875$. Compared, on the other hand, with the implications drawn in Example \ref{exa:AdvantageousnessCCP}, Tab. \ref{tab:Laplace} for a complete graph with $K=2$ classes of derivatives, we obviously see that the risk profile of the described two-tier market with $N=6$ participants is completely different. 
\end{example}
The last example demonstrates that the exact market structure is essential for giving a comprehensive answer to the raised question.

\section{Auxiliary Results for the Application}
\label{sec:AuxiliaryResults}
In this section we provide auxiliary results that can help to clear hurdles related to the application of the network model presented in section \ref{sec:Model}. To the extent of our knowledge, Proposition \ref{prop:Hilbert} about Hilbert transforms is not yet known.\footnote{For a situation where poles inside the contour do show up please refer to section 22.10 in \cite{King2009a}.} The essence of Proposition \ref{prop:AnalyticSignal} is known from \textit{signal processing}, but it has not yet been applied in the context of counterparty or systemic risk.

\subsection{Taking the Maximum and the Hilbert Transform}
\label{subsec:HilbertTransform}
In order to tackle the problem of deriving the Hilbert transform we generalise the result in section 3.4 in \cite{King2009}.
We show how to derive two very useful formulas by using complex analysis and the well-known residue theorem.\\

In order to take the maximum $\max[Y; 0]$ between the sum $Y=\sum_{\lambda \in \Lambda_v}{{}^\pm|X_\lambda|}$ and zero, we have to deal with the Hilbert transform as introduced in (\ref{form:defHT}). Let $\phi_Y$ be a real-valued function, and let the function $\IC \ni z \mapsto \frac{\phi_Y(z)}{\omega -z}$ be extended into the complex plane and bounded by $C$. This extended function is required to be analytic within the complex upper half-plane, except for a finite number of poles $a_1, \ldots, a_m\in \IC$ of order $n_1, \ldots, n_m \in \IN$. Further we assume that $\phi_Y(z) \rightarrow 0$ as $|z| \rightarrow \infty$.  

\begin{prop}
\label{prop:Hilbert}
If the previously stated prerequisites are met then the equations
\begin{align}
     \label{form:ResidueFormula}
     \Hilbert\{\phi_Y(t)\}(\omega) \ = \ 2\ii \sum_{j=1}^{m}{\Res\left[\frac{\phi_Y(z)}{\omega-z}, a_j\right]} 
                                   + \ii \Res\left[\frac{\phi_Y(z)}{\omega-z}, \omega\right].
\end{align}
and
\begin{align}
     \label{form:SimplyfiedResidueFormula}
     \Hilbert\{\phi_Y(t)\}(\omega) \ = \ 2\ii \sum_{j=1}^{m}{\left(\frac{1}{(n_j-1)!} \lim_{z \rightarrow a_j} \frac{\partial^{n_j-1}}{\partial z^{n_j-1}}[(z-a_j)^{n_j} \frac{\phi_Y(z)}{\omega - z}]\right)} - \ii \lim_{z \rightarrow \omega}  \phi_Y(z)
\end{align}
are valid.
\end{prop}

\begin{proof}
The interval $]-R, R[$, as part of the domain of the Hilbert transform (\ref{form:defHT}), is incorporated into the closed path $C= \ C_R \ \cupdot \ ]-R,\omega-\epsilon[ \ \cupdot \ C_\epsilon \ \cupdot \ ]\omega+\epsilon,R[$ as sketched in Fig. \ref{Fig:ContourIntegral}. Obviously, we have $R, \epsilon \in ]0, \infty[$. The positively oriented contour $C$ consists of the semicircle $C_R$, the semicircle $C_\epsilon$ and the segments on the real line.\\
\begin{figure}
  \centering
\tikzstyle{grayStyle}=[rectangle,draw=black!100,fill=gray!80,inner sep=0pt,minimum size=1mm]
\begin{tikzpicture}[scale=3.0]
%\draw[step=.2cm, draw=gray!40] (-1.0,0) grid (1.0,1.0);
\draw [<->, shorten <=1pt,>=stealth',semithick] (-1.2,0) -- (1.2,0);
\draw [<->, shorten <=1pt,>=stealth',semithick] (0,-0.1) -- (0,1.1);
\draw [->, thick] (1,0) arc (0:90:1cm);
\draw [->, thick] (0,1) arc (90:180:1cm);
\draw [->, thick] (-1,0) -- (0.4, 0);
\draw [<-, thick] (0.6,0.0) arc (0:180:0.1cm);
\draw [->, thick] (0.6,0) -- (1,0);
\draw (1.1, 0.0) node[above] {\footnotesize $\IR$};	
\draw (0.0,1.1) node[right] {\footnotesize $\ii \IR$};	
% Lorenzpunkte:
\node[grayStyle] (t0) at (0.5,0.0) {};
\node [below] at (t0) {$t=\omega$};
\node[grayStyle] (a3) at (-0.2,0.3) {};
\node [above] at (a3) {$a_3$};
\node[grayStyle] (a2) at (-0.4,0.5) {};
\node [above] at (a2) {$a_2$};
\node[grayStyle] (a1) at (-0.8,0.25) {};
\node [above] at (a1) {$a_1$};
\node[grayStyle] (a5) at (0.5,0.45) {};
\node [above] at (a5) {$a_{m}$};
\node[grayStyle] (a6) at (0.3,0.6) {};
\node [above] at (a6) {$a_{m-1}$};
\node [below] at (0.1,0.5) {$\ldots$};
\node [above] at (0.8, 0.7) {\bfseries $C_R$};
\node [above] at (0.5, 0.1) {\bfseries $C_\epsilon$};
\node [below] at (-1,0) {\bfseries $-R$};
\node [below] at (1,0) {\bfseries $R$};
\end{tikzpicture}
  	\caption{Contour $C$ and the poles of $\phi_Y(t)$ and $\frac{1}{\omega-t}$}
  \label{Fig:ContourIntegral}  	
\end{figure}
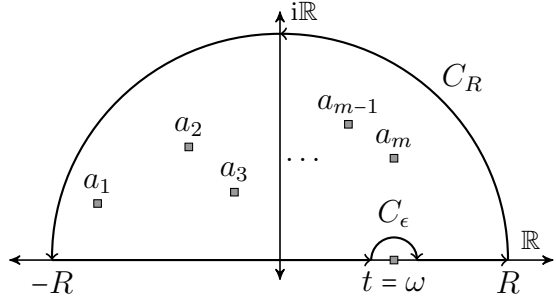

The real-valued integrand is then extended into the complex region bounded by $C$ and $\frac{\phi_Y(z)}{\omega -z}$ is required to be analytic\footnote{See Chapter 1, Holomorphic Functions in \cite{Remmert1991}.} within the complex upper half-plane, except for a finite number of poles $a_1, \ldots, a_m\in \IC$ of order $n_1, \ldots, n_m \in \IN$. Furthermore, we assume that $\phi_Y(z) \rightarrow 0$ as $z \rightarrow \infty$ and we obviously have a simple pole on the real line at $t=\omega$. We can choose $R$ large and $\epsilon$ small enough such that the poles of $\phi_Y$ of the upper half-plane lie within the contour $C$ and do not intersect with the semicircle $C_\epsilon$. \ Applying the residue theorem we get
\begin{align}
     \label{form:ContourIntegral}
 \oint_{C}{\frac{\phi_Y(z)}{\omega-z} dz} \ = \ 2\pi \ii \sum_{j=1}^{m}{\Res\left[\frac{\phi_Y(z)}{\omega-z}, a_j\right]}.
\end{align}
This contour integral can be decomposed into
\begin{align}
  \label{form:ContourIntegralDecomposed}
  \oint_{C}{\frac{\phi_Y(z)}{\omega-z} dz} 
             \ = \  \int_{-R}^{\omega-\epsilon}{\frac{\phi_Y(t)}{\omega-t} dt} 
                    + \oint_{C_\epsilon}{\frac{\phi_Y(z)}{\omega-z} dz} 
                    + \int_{\omega+\epsilon}^{R}{\frac{\phi_Y(t)}{\omega-t} dt}
                    +\oint_{C_R}{\frac{\phi_Y(z)}{\omega-z} dz} .
\end{align}
The first and third integral on the right hand side of equation (\ref{form:ContourIntegralDecomposed}) equals the principle value integral as $R \rightarrow \infty$ and $\epsilon \rightarrow 0$, i.e.,
\begin{align*}
  PV \int_{-\infty}^{\infty}{\frac{\phi_Y(t)}{\omega-t} dt} \
    = \ \lim_{R \rightarrow \infty} \lim_{\epsilon \rightarrow 0}  \left( \int_{-R}^{\omega-\epsilon}{\frac{\phi_Y(t)}{\omega-t} dt} + \int_{\omega+\epsilon}^{R}{\frac{\phi_Y(t)}{\omega-t} dt} \right).
\end{align*}
Let $C_\epsilon$ be parameterised by $\omega + \epsilon e^{\ii \theta}$ with $-\pi \leq \theta \leq 0$, then we obtain
\begin{align*}
  \oint_{C_\epsilon}{\frac{\phi_Y(z)}{\omega-z} dz} \
    = \ \lim_{\epsilon \rightarrow 0}  \left( \ii \int_{-\pi}^{0}{\frac{\phi_Y(\omega + \epsilon e^{\ii \theta})\epsilon e^{\ii \theta}}{\omega-(\omega + \epsilon e^{\ii \theta})} \ d\theta } \right) \
    = \ \pi\ii \phi_Y(\omega) \ = \ -\pi\ii \Res\left[\frac{\phi_Y(z)}{\omega-z}, \omega\right].
\end{align*}
For the last integral we parameterise $C_R$ by $R e^{\ii \theta}$ with $0 \leq \theta \leq \pi$ and we bear in mind that $\phi_Y(z) \rightarrow 0$ as $|z| \rightarrow \infty$. We then deduce
\begin{align*}
  \oint_{C_R}{\frac{\phi_Y(z)}{\omega-z} dz} \
    = \ \lim_{R \rightarrow \infty}  \left( \ii \int_{0}^{\pi}{\frac{\phi_Y(R e^{\ii \theta}) R e^{\ii \theta}}{\omega - R e^{\ii \theta}} \ d\theta } \right) \ = \ 0.
\end{align*}
Equation (\ref{form:ContourIntegral}) therefore simplifies to
\begin{align*}
     PV \int_{-\infty}^{\infty}{\frac{\phi_Y(t)}{\omega-t} dt} \ = \ 2\pi \ii \sum_{j=1}^{m}{\Res\left[\frac{\phi_Y(z)}{\omega-z}, a_j\right]} 
                                                                     + \pi\ii \Res\left[\frac{\phi_Y(z)}{\omega-z}, \omega\right],
\end{align*}
which leads to
\begin{align*}
     \Hilbert\{\phi_Y(t)\}(\omega) \ = \ 2\ii \sum_{j=1}^{m}{\Res\left[\frac{\phi_Y(z)}{\omega-z}, a_j\right]} 
                                   + \ii \Res\left[\frac{\phi_Y(z)}{\omega-z}, \omega\right].
\end{align*}
If we then apply rule 1) and 2) of chapter 13 in \cite{Remmert1991} we further receive formula (\ref{form:SimplyfiedResidueFormula}).
\end{proof}
The two equations (\ref{form:ResidueFormula}) and (\ref{form:SimplyfiedResidueFormula}) are particularly useful for calculating the Hilbert transform of intricate functions as we can see in the next example.%
\begin{example}
\label{exa:HilbertTransform:GammaDistr}
Let $\phi_{|X|}(t)= \frac{1}{1-2\ii t} =\frac{1}{1+4t^2}+\frac{2\ii t}{1+4t^2}$ be the c.f. of a r.v. $|X|$ that is distributed by the \textit{Gamma distribution} $\Gamma(\alpha, \beta)$ with $\alpha:=1$ and $\beta:=2$. Here $\alpha>0$ is the shape parameter and $\beta>0$ the scale parameter. Be aware that the Gamma distribution has only positive samples and that the function $\phi_{|X|}(z)$ converges towards zero if $|z|\rightarrow \infty$. Furthermore, the function $\frac{\phi_{|X|}(z)}{\omega - z}$ is analytic in the upper half-plane except for the pole $-\frac{\ii}{2}$. Applying equation (\ref{form:ResidueFormula}) we get $\Hilbert\{\phi_{|X|}(t)\}(\omega) = \ii \Res\left[\frac{\phi_{|X|}(z)}{\omega-z}, \omega\right] = \frac{2\omega}{1+4\omega^2}-\ii \frac{1}{1+4\omega^2}$.\\ 

Let us now consider the associated negative absolute value $-(|X|)=-|X|$ with $\phi_{-|X|}(t)=\frac{1}{1+4t^2}-\frac{2\ii t}{1+4t^2}$. We now have $\frac{\ii}{2}$ as simple pole. Applying again equation (\ref{form:ResidueFormula}), we receive $\Hilbert\{\phi_{-|X|}(t)\}(\omega) = 2\ii \Res\left[\frac{\phi_{-|X|}(z)}{\omega-z}, \frac{\ii}{2} \right] + \ii \Res\left[\frac{\phi_{-|X|}(z)}{\omega-z}, \omega\right] = \frac{2\omega}{1+4\omega^2} + \ii \frac{1}{1+4\omega^2}$. 
\end{example}
If we consider the last example carefully, we will find several interesting connections between the positive and the negative absolute value as well as to its Hilbert transforms. We will examine these connections in the next section. The introduced notation of this section shall be valid for the entire article.

\subsection{Positive and Negative Absolute Values}
\label{subsec:PositiveNegativePart}
Let $X \sim P$ be a r.v. that is symmetric around the origin and $\phi_X$ its characteristic function. In this section we study the connection between the c.f.s $\phi_X$, $\phi_{|X|}$ and $\phi_{-|X|}$. The latter two c.f.s are used for the representation of the direction of an arrow of the digraph $D$. \\

It is well-known that a c.f. $\phi_X$ of a r.v. $X$ is \textit{Hermitian}, i.e., that $\phi_X(-t) \ = \ \overline{\phi_X(t)}$ for all $t\in \IR$. Due to Theorem 3.1.2 in \cite{Lukacs1970} the r.v. $X$ is \textit{symmetric} if  and only if its c.f. $\phi_X$ is real-valued and even. However, the c.f.s of the positive absolute value $|X|$ as well as of the negative absolute value $-|X|$ are in general complex-valued. We need to take into account the equation $-(|X|) = -|X|$ as well as that $\phi_X$ is Hermitian in order to receive 
\begin{align}
 \label{form:PositivePart}
 \phi_{|X|}(t) \ = \ \eta(t)+\ii \nu(t) 
\end{align}
for the positive absolute and 
\begin{align}
 \label{form:NegativePart}
 \phi_{-|X|}(t) \ = \ \eta(t)-\ii \nu(t)
\end{align}
for the negative absolute value for all $t\in \IR$. Please note that $\eta=\Re(\phi_{|X|})$, $\nu =\Im(\phi_{|X|})$ and that $\phi_{|X|} = \overline{\phi_{-|X|}}$ is obviously valid. In the following we explain how the functions $\eta$ and $\nu$ are connected to each other.\ We call any complex c.f. $\phi(t)$ whose real and imaginary components satisfy the equation
\begin{align}
     \label{def:AnalyticSignal}
 \phi(t) \ = \ \eta(t)+\ii \nu(t) \ = \ \eta(t)+ \ii \Hilbert\{ \eta \}(t)   \qquad t\in \IR
\end{align}
an \textit{analytic signal}\footnote{See section 4.1.4 in \cite{King2009}.}.

\begin{prop}
\label{prop:AnalyticSignal}
The c.f. $\phi_{|X|}$ of a positive absolute value $|X|$ with $X\sim P$ is an analytic signal in a natural way\footnote{Compare with section 18.4 in \cite{King2009a}.}, i.e.,
\begin{align}
     \label{form:PosPartAnalyticSignal}
    \phi_{|X|}(t) \ = \ \phi_X(t)+ \ii \Hilbert \{ \phi_X \}(t).
\end{align} 
\end{prop}
\begin{proof}
Let $X$ be a symmetric r.v. with real-valued and even c.f. $\phi_X$ and let us further assume that $\phi_X$ is absolutely integrable. Then, if its Fourier transform\footnote{See section 2.6.1 in \cite{King2009}.} $f := \ \calF\{\phi_X\}$ is also absolutely integrable, we can use the \textit{inverse Fourier transform} 
\begin{align}
     \label{form:InvFourier}
	\calF^{-1}\{f\}(t) = \ \frac{1}{2\pi} \int_{-\infty}^{\infty}{f(x) e^{\ii x t} \, dx} \ = \ \phi_X(t)
\end{align}
to recover the input function $\phi_X$ from its Fourier transform. We can imagine $f$ as a density function that represents the distribution $P$. The definition of $|X|$ basically means that all negative samples are rejected and the probability for the positive samples is doubled. Thus, we set $f^+(x):= f(x) + \sgn(x)f(x)$ for all $x\in \IR$, where 
\begin{align*}
     \sgn(x)= \begin{cases}
                    1, & \text{ for } x>0 \\
                    0, & \text{ for } x=0 \\
                    -1, & \text{ for } x<0 
     \end{cases}
\end{align*}
is the \textit{signum function}. The function $f^+$ is zero for all negative and $2 f(t)$ for all positive real numbers. Apparently, the function $f^+$ is a representation for the positive absolute value $|X|$. Applying the inverse Fourier transform to $f^+$ we get
\begin{align*}
     \calF^{-1}\{f^+\} \ = \ \calF^{-1}\{f\} + \calF^{-1}\{\sgn \cdot f\} \  = \ \phi_{X} + \calF^{-1}\{\sgn \cdot f\} \ = \ \phi_{|X|}
\end{align*}
because of formula (\ref{form:InvFourier}) and the additivity of the integral. Due to equation (\ref{form:PositivePart}) we further know that $\calF^{-1}\{f^+\} = \eta + \ii \nu$. Comparison of real and imaginary parts yields to $\calF^{-1}(\sgn \cdot f)=\ii \nu$, which means that $-\ii \cdot \sgn \cdot f$ is related  to $\nu$ by the inverse Fourier transform. According to equation (5.2) in \cite{King2009} the inverse Fourier transform of $-\ii \cdot \sgn(x)$ equals $\frac{1}{\pi x}$ and because of (4.154) in \cite{King2009} we receive
\begin{align*}
     \nu(t) \ = \ \phi_X(x) * \frac{1}{\pi x} = \frac{1}{\pi} \int_{- \infty}^{\infty}{ \frac{\phi_X(t)}{x-t} \, dx} = \Hilbert \{\phi_X\}(t), 
\end{align*}
where $*$ is the convolution\footnote{Section 2.6.2 in \cite{King2009}.}. Putting the parts together to 
\begin{align}
    \phi_{|X|}(t) \ = \ \phi_X(t)+ \ii \Hilbert \{ \phi_X \}(t),
\end{align}
we easily infer that the c.f. $\phi_{|X|}$ of the positive absolute value $|X|$ is an analytic signal. This basically means that the negative samples of the real-valued distribution $P$ are superfluous in this context. An illustrative explanation for this fact provides the symmetry of the distribution $P$.
\end{proof}

Because of $\phi_{-|X|}(t) = \overline{\phi_{|X|}(t)}$ it follows
\begin{align}
     \label{form:NegPartAnalyticSignal}
    \phi_{-|X|}(t) \ = \ \phi_X(t) - \ii \Hilbert \{\phi_X\}(t).
\end{align}
Furthermore, the $n$-th power of the Hilbert transform of the analytic signal $\phi_{|X|}$ can be written as 
\begin{align}
     \label{form:HilbertTransformOfPositiveParts}
     \Hilbert\{\phi_{|X|}^n\} \ = \  -\ii  \phi_{|X|}^n, \qquad n\in \IN,
\end{align}
due to equation (4.252) in \cite{King2009}.\\

Formulas (\ref{form:PosPartAnalyticSignal}) and (\ref{form:NegPartAnalyticSignal}) show how the c.f. of a symmetric distribution $P$ and the c.f. of its positive and negative absolute values are connected to each other by the Hilbert transform. Both formulas are also useful when we employ a distribution with only positive or negative outcomes as positive or negative absolute values\footnote{Please refer to Example \ref{exa:HilbertTransform:GammaDistr}.}. We can then use formulas (\ref{form:PosPartAnalyticSignal}) and (\ref{form:NegPartAnalyticSignal}) to get the Hilbert transform simply by considering the imaginary component of the positive or negative absolute value of the distribution.

\begin{example}
We consider a netting set $\Lambda_v$ comprising $m\in \IN$ positive trade positions of a counterparty $v$. Each position is represented by a i.i.d. r.v. $|X_i| \sim \Gamma(\alpha,\beta)$ with $i\in \{1, 2, \ldots, m\}$. The absolutely integrable c.f. of the positive absolute value $|X_i|$ is $(1-\beta\ii t)^{-\alpha}$ and the sum $Y:= \sum_{i=1}^{m}{|X_i|}$ is determined by the c.f. $(1-\beta\ii t)^{-\alpha m}$. The sum $Y$ is  distributed by $\Gamma(m\alpha, \beta)$ and its Hilbert transform equals $-\ii (1-\beta\ii t)^{-\alpha m}$, because of equation (\ref{form:HilbertTransformOfPositiveParts}). Please also compare these results to Example \ref{exa:HilbertTransform:GammaDistr}. It is quite straightforward to apply formula (\ref{form:CharFuncMoments}) to get the expectation $\alpha\beta m$ of the credit exposure. The situation becomes more interesting, when the netting set contains positive \textit{and} negative trade positions of the counterparty $v$. Then, we have to take the maximum between the sum of the r.v.s and zero in order to determine the expected counterparty credit risk.
\end{example}

\section{Conclusion}
We endorse the view of several authors that considering the precise market structure for studying counterparty credit risk or systemic risk is essential. We provide a new type of network model which is capable of capturing the precise structure of any given financial market based, for example, on empirical findings. With the attached stochastic framework it is further possible to study how a network structure and counterparty credit risk are connected to each other. This allows us to study different structures and their characteristics relating to, for instance, systemic risk. We show that Eulerian digraphs are distinguished exposure structures in the context of counterparty risk and we reveal that different structures can have a significantly different impact on the overall risk. We therefore suggest that the individual structure of a financial market should be taken into consideration.\\

We use the powerful theory of characteristic functions as well as the theory of Hilbert transforms. Deriving the specific characteristic function as well as its Hilbert transform can be a great challenge. However, we provide useful insight into both concepts in order to overcome these barriers in many cases. \ The model presented here is quite flexible and could be easily modified to meet specific requirements. For example, it could  be used to study the structure of counterparty credit risk within other types of markets, different netting rules and more complex distributions such as extreme value distributions. One could also use the model to study analytically how shocks affect a specific network by changing the distribution (parameter) in an appropriate way.

\section{Proofs}
This section contains the proofs of the two structure theorems.

\subsection{Proof of Structure Theorem \ref{theorem:CounterpartyRiskDigraph}}
\label{subsec:ProofDigraphTheorem}
Assume $Y_v := \sum_{\lambda \in \Lambda_v}{X_\lambda}$. 
\begin{enumerate}
	\item[(i)] Let $\lambda \in \Lambda_v$ and let further $\Lambda_v^+$ and $\Lambda_v^-$ be the sets with $h(\lambda) = v$ and $t(\lambda) = v$,
                respectively. The two sets form a partition of $\Lambda_v$, that is, $\Lambda_v = \Lambda_v^+ \cupdot \Lambda_v^-$. 
                Be aware that $|\Lambda_v^+|=\gamma_+(v)$ and $|\Lambda_v^-|=\gamma_-(v)$.
                Keeping in mind the linearity of the conditional expectation, we deduce from 
                $0=\IE(Y_v)=\IE \left(\sum_{\lambda \in \Lambda_v}{{}^\pm|X_{\lambda}|} \right)$ the equivalent equation 
                $0 = \sum_{\lambda\in \Lambda_v^+}{\IE(|X_{\lambda}|}) + \sum_{\lambda'\in \Lambda_v^-}{\IE(-|X_{\lambda'}|)}$. 
                Because each r.v. follows the same symmetric distribution around $0$, we receive the validity of 
                $0 = |\Lambda_v^+| \cdot \IE(|X_{\lambda}|) + |\Lambda_v^-| \cdot \IE(-|X_{\lambda'}|)$. 
                The equation $\IE(|X_{\lambda}|)=-\IE(-|X_{\lambda'}|)$ is valid because of the symmetry of the r.v. and therefore we obtain 
                $0 = [\gamma_+(v) -\gamma_-(v)] \cdot \IE(|X_{\lambda}|)$. We conclude that $\gamma(v)=0$, because $\IE(|X_{\lambda}|)>0$. If we assume 
                $\gamma(v)=0$ we can use the same arguments to show that the equation $\IE(\sum_{\lambda \in \Lambda_v}{{}^\pm|X_\lambda|})=0$ is valid. 
               
	\item[(ii)] The c.f. of $\max[\sum_{\lambda \in \Lambda_v}{{}^\pm|X_\lambda|};0]= \max[Y_v;0]$ is given by (\ref{form:MaxOfRV}).  
                 Applying formula (\ref{form:CharFuncMoments}) we receive the corresponding expectation 
                 $\IE\left(\max[\sum_{\lambda \in \Lambda_v}{{}^\pm|X_\lambda|};0]\right) = 
                 \frac{1}{2}\IE(Y_v) + \frac{1}{2}\partial_t[\Hilbert\{ \phi_{Y_v}\}(t) - \Hilbert \{\phi_{Y_v}\}(0)](0)$. Because of (i) we have got 
                 $\IE(Y_v)=0$ if and only if $\gamma(v)=0$. Let us now assume that $\gamma(v)=0$ then 
                 $\gamma_+(v)=|\Lambda_v^+|=|\Lambda_v^-|=\gamma_-(v)$. According to equations (\ref{form:PositivePart}) and (\ref{form:NegativePart})
                 the product $\phi_{Y_v} = \left(\prod_{\lambda \in \Lambda_v^+}{\phi_{|X_\lambda|}}\right) 
                 \left(\prod_{\lambda'\in \Lambda_v^-}{\phi_{-|X_\lambda'|}}\right)$ must be real-valued and even. 
                 The parity property\footnote{See (4.5) and (4.6) in \cite{King2009}.} of the Hilbert transform implies that $\Hilbert \{\phi_{Y_v}\}$
                 is an odd continuous function and therefore $\Hilbert \{\phi_{Y_v}\}(0)=0$.
\end{enumerate}

\subsection{Proof of Structure Theorem \ref{theorem:CounterpartyRiskGraph}}
\label{subsec:ProofGraphTheorem}
The c.f. $\phi_X$ of a r.v. $X \sim P$ is real-valued and even and so is the c.f. $\phi_Y$ of the sum $Y:=\sum_{\lambda \in \Lambda_v}{X_\lambda}$. Because of the parity property the Hilbert transform $\Hilbert \{\phi_Y(t)\}(t)$ is a continuous odd function. Since $\Hilbert \{\phi_Y\}(0)=0$, we obtain $\phi_{\max[Y;0]}(t) = \frac{1}{2}\left[1 + \phi_Y(t) \right] + \frac{\ii}{2}\left[ \Hilbert \{\phi_Y(t)\}(t) \right]$ and therefore the expectation can be calculated by $\frac{\partial_t [ \phi_{\max[Y;0]}(t)](0)}{\ii}   \ = \ \frac{1}{2\ii} \partial_t [\phi_Y(t)](0)  + \frac{1}{2}  \partial_t\left[\Hilbert\{\phi_Y(t)\}(t) \right](0)$. \\  

The term $\frac{\partial_t [\phi_Y(t)](0)}{\ii}$ stands for the expectation of the sum of $|\Lambda_v|$ i.i.d. random variables with zero mean. Thus, $\frac{\partial_t [\phi_Y(t)](0)}{\ii} =0$ and this leads us to the equation
\begin{align*}
    \IE\left(\max[\sum_{\lambda \in \Lambda_v}{X_\lambda}; 0]\right) = \frac{1}{2} \partial_t[\Hilbert\{\phi_Y(t)\}](0)
\end{align*}
for an arbitrary vertex $v\in V$ of the graph $G$.
\newpage
\printbibliography
\end{document}